\documentclass[11pt]{article}
\usepackage[centertags]{amsmath}
\usepackage{mathtools}
\usepackage{amssymb}
\usepackage{amsthm,color}
\usepackage{graphicx}
\usepackage{upquote}
\usepackage{cases}
\usepackage{scalerel}
\newcommand\reallywidehat[1]{\arraycolsep=0pt\relax%
\begin{array}{c}
\stretchto{
  \scaleto{
    \scalerel*[\widthof{\ensuremath{#1}}]{\kern-.5pt\bigwedge\kern-.5pt}
    {\rule[-\textheight/2]{1ex}{\textheight}} 
  }{\textheight} %
}{0.5ex}\\           
#1\\                 
\rule{-1ex}{0ex}
\end{array}
}

\newtheorem{theorem}{Theorem}[section]
\newtheorem{lemma}[theorem]{Lemma}
\newtheorem{cor}[theorem]{Corollary}
\newtheorem{proposition}[theorem]{Proposition}
\newtheorem{defi}{Definition}[section]
\newtheorem*{rem}{Remark}
\theoremstyle{remark}
\newtheorem{example}{Example}[section]

\def\={\hspace{-3mm}&=&\hspace{-3mm}}

\title{The 
Borel transform and linear nonlocal equations: applications to zeta-nonlocal field models}
\author{A. Ch\'avez$^1$\footnote{E-mail: alancallayuc@gmail.com}~, H. Prado$^2$\footnote{E-mail: humberto.prado@usach.cl}~,
and E.G. Reyes$^3$\footnote{E-mail: e\_g\_reyes@yahoo.ca ; enrique.reyes@usach.cl} \\
\small{$^{1}$Departamento de Matem\'aticas,} \\
\small{ Universidad Nacional de Trujillo,}\\
\small{Av. Juan Pablo II s/n. Trujillo-Per\'u}\\
\small{$^{2,3}$ Departamento de Matem\'atica y Ciencia de la
Computaci\'on,} \\
\small{ Universidad de Santiago de Chile }\\
\small{Casilla 307 Correo 2, Santiago, Chile }}

\begin{document}

\maketitle

\begin{abstract}
We define rigorously operators of the form $f(\partial_t)$, in which $f$ is an analytic function
on a simply connected domain. Our formalism is based on the Borel transform on entire functions of exponential type.
We study existence and regularity of real-valued solutions for the nonlocal in time equation
\begin{equation*}
f(\partial_t) \phi = J(t) \; \; , \quad t\in \mathbb{R}\; ,
\end{equation*} 
and we find its more general solution as a restriction to $\mathbb{R}$ of an entire function of exponential type.
As an important special case, we solve explicitly the linear nonlocal zeta field equation
\begin{equation*}
\zeta(\partial_t^2+h)\phi = J(t)\; ,
\end{equation*}
in which $h$ is a real parameter, $\zeta$ is the Riemann zeta function, and $J$ is an entire function of exponential type.
We also analyse the case in which $J$ is a more general analytic function (subject to some weak technical assumptions).
This case turns out to be rather delicate: we need to re-interpret the symbol $\zeta(\partial_t^2+h)$ and to leave
the class of functions of exponential type. We prove that in this case the zeta-nonlocal equation above admits an
analytic solution on a Runge domain determined by $J$. This

The linear zeta field equation is a linear version of a field model depending on the Riemann zeta function arising
from $p$-adic string theory.
\end{abstract}

\maketitle

\section{Introduction}\label{sec:Introduction}
\setcounter{equation}{0}
In this paper a {\em nonlocal operator} is an expression of the form $f(\partial_t)$, in which $f$ is an analytic function,
and a {\em nonlocal equation} is an equation in which a nonlocal operator appears. It is of course well-known how to
define the action of $f(\partial_t)$ on a given class of functions
if the ``the symbol" $f$ is a polynomial. It is not obvious how to extend this
definition to more general symbols $f$: for instance, $f$ may be
beyond the reach of classical tools used in the study of
pseudo-differential operators ({\em e.g.}, the derivatives of $f$ may not satisfy appropriated bounds, see \cite{H,DubinskyBook,U}).
We provide a rigorous definition of $f(\partial_t)$ for
a large class of analytic functions $f$ ---which includes the Riemann
zeta function $\zeta$--- in the main body of this work. The theory presented herein includes
symbols which cannot be considered via Laplace transform (as in our previous paper
\cite{CPR_Laplace}). Our main example of such a symbol is $\zeta(\partial_t^2+h)$, see
\cite[Section 6]{CPR_Laplace} and Section 2 below.

\medskip

We frame our discussion within classical analytic function theory. Interestingly, nonlocal
operators appear naturally in this abstract mathematical framework, for example, in the study of
the distribution of zeroes of entire functions, see \cite{CardonDA02,CardonGas05,LevinBook1} and
references therein. We present one result as an illustration:

The Laguerre-P\'olya class, denoted by $\mathcal{L}\mathcal{P}$, is defined as the collection of
entire functions $f$ having only real zeros, and such that $f$ has the following factorization
\cite[sections 2.6 and 2.7]{BoasBook}:
$$
f(z)=cz^m e^{\alpha z- \beta z^2} \prod _{k}\left(1-\dfrac{z}{\alpha_k} \right)e^{z/\alpha_k}\; ,
$$
where $c,\alpha, \beta, \alpha_k $ are real numbers, $\beta \geq 0$, $\alpha_k \not = 0$, $m$ is
a non-negative integer, and $\sum_{k=1}^{\infty}\alpha_k^{-2}< \infty$.

Let $D$ be the differentiation operator and $\phi \in \mathcal{L}\mathcal{P}$; the following
lemma presents one important instance in which the nonlocal operator $\phi(D)$ (formally defined
via power series, see \cite{LevinBook1} and Subsection 3.2 below) is well defined, see
\cite[Theorem 8, p. 360]{LevinBook1}.

\begin{lemma}\label{lemma0}
Let $\phi, f \in \mathcal{L}\mathcal{P}$ such that
$$\phi(z)= e^{-\alpha z^2}\phi_1(z) {\mbox \; and \; } f(z)= e^{-\beta z^2}f_1(z)\; ,$$
where $\phi_1,f_1$ have genus $0$ or $1$ and
$\alpha, \beta \geq 0$. If $\alpha \beta <1/4$, then $\phi(D) f \in \mathcal{L}\mathcal{P}$.
\end{lemma}
\noindent The notion of the genus of a function is explained in \cite[p. 22]{BoasBook}. We have,
see \cite[Theorem 1]{CardonGas05}.
\begin{theorem} \label{ef}
Let $\phi, f \in \mathcal{L}\mathcal{P}$ such that
$$\phi(z)= e^{-\alpha z^2}\phi_1(z) {\mbox \; and \; } f(z)= e^{-\beta z^2}f_1(z) \; ,$$
where  $\phi_1,f_1$ have genus $0$ or $1$
and $\alpha, \beta \geq 0$. If $\alpha \beta <1/4$ and $\phi$ has infinitely many zeros, then $\phi(D) f$ has only simple
and real zeros.
\end{theorem}
%


We will not use this result explicitly in this paper, but zeroes of entire functions {\em are}
crucial in our theory, see for instance Theorem 3.17 and Example 4.1 below, and so we expect that
results such as Theorem \ref{ef} will play a role in future developments. Now, as we stated in
\cite{CPR_Laplace}, our main motivation for the study of nonlocal equations comes from Physics.
Nonlocal operators and equations appear in contemporary physical theories such as string theory,
cosmology and non-local theories of gravity. We refer the reader to
\cite{BK1,BK2,BiswasTir15,GianlucaBook1,LeiFen17} for information on the last two topics. With
respect to string theory, let us mention two important equations:
\begin{equation} \label{padic-1}
p^{a\, \partial^2_t}\phi = \phi^p \; , \; \; \; \; \; a > 0 \; ,
\end{equation}
and
\begin{equation} \label{padic-101}
e^{\, \partial^2_t/4}\phi = \phi^{p^2}\; .
\end{equation}
Equation (\ref{padic-1}) describes the dynamics of the open $p$-adic string for the scalar
tachyon field, while equation (\ref{padic-101}) describes the dynamics of the closed $p$-adic
string for the scalar tachyon field, see {\it e.g.} \cite{AV,BK2,D,Moeller,V,VV,Vla3}.
These equations have been studied formally via integral equations and a ``heat equation" method,
see {\em e.g.} \cite{AV,Moeller,MN,V,VV,Vla4}.

In the intriguing paper \cite{D}, B. Dragovich constructs  a field theory starting from
(\ref{padic-1}) whose equation of motion (in $1+0$ dimension) is
\begin{equation} \label{zeta}
\zeta \left(- \frac{1}{m^2} \, \partial_t^2 + h \right) \phi = U(\phi) \; .
\end{equation}
Here $\zeta$ is the Riemann zeta function as we noted before; $m,h$ are real parameters and
$U$ is some nonlinear function. We understand
the Riemann zeta function as the analytic extension of the infinite series
$$\zeta(s)=\sum_{n=1}^{\infty}\frac{1}{n^s},\; \; Re(s)>1\; ,$$
to the whole complex plane, except $s=1$, where it has a pole of order $1.$
The main aim of this work is the development of a general theory (see section 3) to solve
linear nonlocal in time equations of the form
\begin{equation}\label{Gen_Eq}
f(\partial_t)\phi(t)=g(t)\; , \; \; \; \; \; t \in  \mathbb{R}\; ,
\end{equation}
in which the {\em symbol} $f$ is an arbitrary analytic function on a simply connected domain.
We have presented a rigorous setting for (a restricted class of) equations (\ref{Gen_Eq}) in
\cite{CPR_Laplace}, using Laplace transform. However, in \cite{CPR_Laplace} we also remark that
the important equation
\begin{equation} \label{zeta1}
\zeta \left(- \frac{1}{m^2} \, \partial_t^2 + h \right) \phi = J(t) \; ,
\end{equation}
naturally motivated by (\ref{zeta}), is beyond the reach of our Laplace transform method.
One way to move forward is to use Borel transform, as in \cite{CPR}. However, the analysis of
\cite{CPR} is restricted to symbols $f$ (see Equation (\ref{Gen_Eq})) which are entire functions,
and to right hand side terms $J$  which are functions of exponential type. Equation (\ref{zeta1})
motivates us to remove these two restrictions. This is precisely what we do in this work.

\smallskip

Let $\Omega \subseteq \mathbb{C}$ be a simply connected domain and,   denote by $Exp(\Omega)$ the space of entire functions
of finite exponential type such that its elements have Borel transform with singularities in $\Omega$ (see section 3). Let
$f$ be an holomorphic function in $\Omega$; we use the Borel transform to define rigorously the operator $f(\partial_t)$ on
the space $Exp(\Omega)$, in a way that evokes the definition of classical
pseudo-differential operators via Fourier transform. Then, using
this theory, we
find the most general solution to Equation (\ref{Gen_Eq}) as a restriction to $\mathbb{R}$ of a
function in $Exp(\Omega)$, and we apply our theory to the following (normalized version of (\ref{zeta1}) with our signature
conventions) linear zeta-nonlocal field equation
\begin{equation}\label{eqzeta2}
\zeta(\partial_t^2+h)\phi = J\; .
\end{equation}

Now we expose the organization of this work. In Section 2 we give some preliminary comments which motivate the present
work. In Section 3 we consider a general analytic symbol $f$ and we
define the action of the operator
$f(\partial_t)$ on the space of entire functions of exponential type using Borel transform. We also solve explicitly
Equation (\ref{Gen_Eq}). In Section 4 we apply the theory developed in Section 3 to the linear zeta-nonlocal scalar field
equation (\ref{eqzeta2}): the zeroes of the Riemann zeta function play an important role in representing its solution. Also
in this section, we introduce the space $\mathcal{L}_{>}(\mathbb{R}_+)$ of all real
functions $g$ with domain $[0,+\infty)$ such that there exist their Laplace transform $\mathcal{L}(g)$, and
$\mathcal{L}(g)$ has an analytic extension to an angular contour, and we study and solve equation (\ref{eqzeta2}) for right
hand side
in $\mathcal{L}_{>}(\mathbb{R}_+)$. This study involves some delicate limit procedures which take us outside the
class of functions of exponential type. Finally, in an appendix we formally derive some equations of motion of interest
from a mathematical point of view, including the zeta-nonlocal scalar field proposed by B. Dragovich.

\section{A preliminary discussion}

As it is discussed in the work of Dragovich \cite{D,D1}, see also the appendix of this work, the following nonlocal
equation
\begin{equation}\label{Zeq_04}
\zeta\left( \dfrac{\square}{2m^2}+h \right )\psi= g(\psi)
\end{equation}
in which $g(z)$ is an analytic non-linear function of $z$,
appears naturally as an interesting mathematical modification of $p$-adic string theory \cite{D,D1,D2}.

Motivated by Equation (\ref{Zeq_04}), we study linear equations in $1+0$ dimensions of the form
\begin{equation} \label{Zeq_044}
\zeta(\partial_t^2+h) \phi = J\; .
\end{equation}
in which we are using a signature so that $\square = \partial_t^2$ simply for comparison purposes with our previous
articles.

In our previous work \cite{CPR_Laplace}, we have studied linear nonlocal equations (and its associated Cauchy problem)
using an approach based on Laplace transforms and the Doetsch representation theorem, see \cite{Doetsch2}:
If $L^p([0,+\infty))$ is the standard $L^p$-Lebesgue space and $H^q(\mathbb{C}_+)$ is the Hardy space, there exists a
correspondence between these spaces determined by the Laplace transform
$\mathcal{L}:L^p([0,+\infty))\to H^q(\mathbb{C}_+)$ for appropriated Lebesgue exponents $p,q$. In this situation,
we obtained exact formulas for the representation of the solution for equations such as (\ref{Gen_Eq}). (The approach
of \cite{CPR_Laplace} supersedes previous work \cite{GPR_JMP,GPR_CQG}). One of our results is the following theorem
(The function $r$ appearing therein in a ``generalized initial condition", see \cite[Section 3]{CPR_Laplace}):

\begin{theorem} \label{main_thm_1}
Let us fix a function $f$ which is analytic in a region $D$
which contains the half-plane $\{s \in \mathbb{C} : Re(s) > 0\}$. We also fix $p$
and $p'$ such that $1<p\leq 2$ and $\frac{1}{p}+\frac{1}{p'}=1$, and we consider a
function $J \in L^{p'}(\mathbb{R}_+)$ such that
$\mathcal{L}(J) \in H^p(\mathbb{C}_+)$.
We assume that the function $(\mathcal{L}( J ) + r)/f$ is in the space
$H^p(\mathbb{C}_+)$. Then, the linear equation
\begin{equation} \label{lin_gen_1}
f(\partial_t)\phi = J
\end{equation}
can be uniquely solved on $L^{p'}(0,\infty)$.
Moreover, the solution is given by the explicit formula
\begin{equation}\label{lin_gen_1111}
\phi = \mathcal{L}^{-1} \left( \frac{\mathcal{L}( J ) + r}{f} \, \right) \; .
\end{equation}
\end{theorem}

\noindent Moreover, using some technical assumptions (see \cite[corollary 2.10]{CPR_Laplace}), the representation
formula (\ref{lin_gen_1111}) for the solution can be reduced to
\begin{equation}\label{sol_lor_11}
\phi = \mathcal{L}^{-1} \left( \frac{\mathcal{L}( J )}{f}\right) + \mathcal{L}^{-1} \left(\frac{r}{f} \, \right) \; .
\end{equation}
The theory can be applied to various field models; in particular, it can be applied to zeta-nonlocal field models of
the form
\begin{equation*} \label{Zeq_0444}
\zeta(\partial_t+h) \phi = J\; .
\end{equation*}
for appropriate functions $J$, see  \cite[Section 4]{CPR_Laplace}.

Now, let us denote by $\mathcal{A}(\mathbb{C}_+)$ the class of functions which are analytic in a region $D$ which contains
the half-plane $\{s \in \mathbb{C} : Re(s) > 0\}$. We can see that for $h>1$ the symbol $\zeta(s+h)$ is in the class
$\mathcal{A}(\mathbb{C}_+)$ while, if $p(s):=s^2$, the symbol $\zeta_h\circ p(s):=\zeta(s^2+h)$ is not, as we explain
presently. It follows from this observation that for some basic forces $J$ (e.g. piecewise smooth functions with
exponential decay) we have $ \frac{\mathcal{L}(J)}{\zeta_h\circ p} \not \in H^p(\mathbb{C}_+)\,$, and therefore the
representation formula  (\ref{sol_lor_11}) of the solution breaks down. We conclude that the study of Equation
(\ref{Zeq_044}) requires a generalization of the theory developed in \cite{CPR_Laplace}.

First of all, we observe that the properties of the Riemann zeta function (see for instance \cite[Section 4]{CPR_Laplace})
imply that the symbol
\begin{equation} \label{nszf}
\zeta(s^2+h)=\sum_{n=0}^{\infty}\dfrac{1}{n^{s^2+h}}
\end{equation}
is analytic in the region $\Gamma:= \{ s \in \mathbb{C} : Re(s)^2-Im(s)^2>1-h \}$, which is not a half-plane; its analytic
extension
$\zeta \circ p$ has poles at the vertices of the hyperbolas $Re(s)^2-Im(s)^2=1-h $, and its critical region is the set
$\{ s \in \mathbb{C} : -h< Re(s)^2-Im(s)^2<1-h \}$. In fact, we recall from \cite[Section 6]{CPR_Laplace} that according to
the value of $h$ we have:
\begin{itemize}
\item [i)]For $h>1$, $\Gamma$ is the region limited by the interior of the dark hyperbola $Re(s)^2-Im(s)^2=1-h $
containing the real axis:
\begin{center}
\includegraphics[width=4cm, height=4cm]{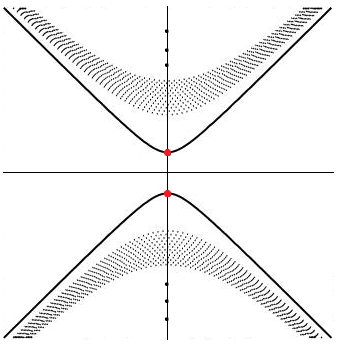}
\end{center}
{\footnotesize
The poles of $\zeta(s^2+h)$ are the vertices of dark hyperbola, indicated by two thick dots. The trivial zeroes of
$\zeta(s^2+h)$ are indicated by thin dots on the imaginary axis; and the non-trivial zeroes are located on the darker
painted region (critical region).}

\item [ii)]For $h<1$, $\Gamma$ is the interior of the dark hyperbola $Re(s)^2-Im(s)^2=1-h $ containing the imaginary axis:
\begin{center}
\includegraphics[width=4cm, height=4cm]{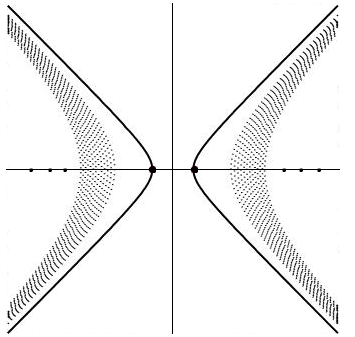}
\end{center}
{\footnotesize
The poles of $\zeta(s^2+h)$ are the vertices of dark hyperbola, indicated by two thick dots. The trivial zeroes of
$\zeta(s^2+h)$ are indicated by thin dots on the real axis; the non-trivial zeroes are located on the darker painted
region (critical region).}

\item [iii)]For $h=1$, $\Gamma$ is the interior of the cones
limited by the curves $y=|x|, y=-|x|$.
\begin{center}
\includegraphics[width=4cm, height=4cm]{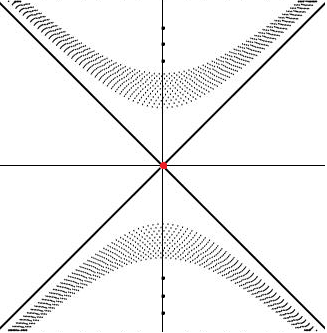}
\end{center}
{\footnotesize The pole of $\zeta(s^2+1)$ is the origin (vertex of
dark curves $y=|x|, y=-|x|$). The trivial zeroes of $\zeta(s^2+h)$
are indicated by thin dots on the imaginary axis; the non-trivial
zeroes are located on the darker painted region (critical
region).}
\end{itemize}

On the other hand, since the Riemann zeta function has an infinite number of nontrivial zeroes in the critical strip (as
famously proven by Hadamard and Hardy, see \cite{KaVo} for original references), then the function
$\zeta_{h}\circ p(\cdot)$  has also an infinite number of nontrivial zeroes on its critical region. We denote by
$\mathcal{Z}$ the set of all such zeroes. By i), ii) and  iii) we have that $\mathcal{Z}$ is contained in the
corresponding dark dotted region. Moreover
$$ \sup_{z\in \mathcal{Z}}|Re(z)|=+\infty.$$
This analysis implies that
the function $ \mathcal{L}(J)/(\zeta_h\circ p)$ does not necessarily belongs to $H^p(\mathbb{C}_+)\,$, and therefore the
expression $\mathcal{L}^{-1}\left( \frac{\mathcal{L}(J)}{\zeta_h\circ p} \right)$ in the representation of the solution
(\ref{sol_lor_11}) does not always make sense.

Thus, a new approach for the study of Equation (\ref{Zeq_04}) is necessary. As stated in Section 1, the method that we use
is based on the Borel transform, see \cite{CPR,U,BoasBook,DubinskyBook} and references therein.

\section{The general theory for nonlocal equations}

\subsection{Entire functions of exponential type}

\begin{defi}
An entire function $\phi:\mathbb{C}\to \mathbb{C}$ is said to be of finite exponential type $\tau_{\phi}$ and finite order
$\rho_{\phi}$ if $\tau_{\phi}$ and $\rho_{\phi}$ are the infimum of the positive numbers $\tau, \rho$ such that the following
inequiality holds:
$$|\phi(z)|\leq Ce^{\tau|z|^{\rho}}, \quad \forall z\in \mathbb{C}\; ,\; {\; and\; some\; } C>0.$$
\end{defi}

\noindent When $\rho_{\phi}=1$, the function ${\phi}$ is said to be of {\bf exponential type}, or of {\bf  exponential type
$\tau_{\phi}$}, if we need to specify its type. If we know the representation of a entire function $\phi$ as a power series,
then a standard way to calculate its order, see \cite[Theorem 2.2.2]{BoasBook}, is by using the formula
\begin{equation}\label{Order}
\rho = \left( 1-\lim_{n\to \infty}\sup\dfrac{\ln{|\phi^{(n)}(0)|}}{n\ln(n)} \right)^{-1}\; ,
\end{equation}
while its type is calculated as follows (see formula 2.2.12, page 11 in \cite{BoasBook}):
\begin{equation}\label{Type}
\sigma=\lim_{n\to \infty}\sup |\phi^{(n)}(0)|^{1/n}\; .
\end{equation}
The space of functions of exponential type will be denoted by $Exp(\mathbb{C})$.

\begin{defi}
Let  $\phi$ be an entire function of exponential type $\tau_{\phi}$. If
 $\phi(z)=\sum_{n=0}^{\infty}a_nz^n$; then, the Borel transform  of $\phi$ is defined by
$$B(\phi)(z):=\sum_{n=0}^{\infty}\dfrac{a_nn!}{z^{n+1}}\; .$$
\end{defi}

It can be checked that $B(\phi)(z)$ converges uniformly for $|z|>\tau_{\phi}$, see \cite[p. 106]{U}, and therefore
it defines an analytic function on $\{z\in \mathbb{C} : |z|>\tau_{\phi}\}$.

An alternative way to calculate the Borel transform of an entire function $\phi$ of exponential type $\tau_{\phi}$  is
the use of the complex Laplace transform, see \cite{BoasBook}: if $z = |z| \exp(i \theta)$ is such that
$|z|=r>\tau_{\phi}$, then
\begin{equation}\label{EqCBT}
B(\phi)(re^{i\theta})=e^{i\theta}\int_0^{\infty}\phi(te^{i\theta})e^{-rt}dt\; .
\end{equation}
\noindent
In particular, if $z\in \mathbb{R}$ is such that $z>\tau_{\phi}$, then $B(\phi)$ can be obtained as the analytic
continuation of its real Laplace transform:
\begin{equation}\label{EqRBT}
\mathcal{L}(\phi)(z)=\int_0^{\infty}\phi(t)e^{-zt}dt\; .
\end{equation}

For $\phi \in Exp(\mathbb{C})$, we let $s(B(\phi))$ denote the set of singularities of the Borel transform of $\phi$,
and we also denote by $S(\phi)$ the {\em conjugate diagram} of $B(\phi)$, this is, the closed convex hull of the set of
singularities $s(B(\phi))$. The set $S(\phi)$ is a convex compact subset of $\mathbb{C}$, and we can check that $B(\phi)$
is an analytic function in $\mathbb{S}\setminus S(\phi)$, where $\mathbb{S}$ is the extended
complex plane $\mathbb{C}\cup \{\infty\}$ and we have set $B(\phi)(\infty)=0$.


\begin{rem} Hereafter we will use the following notation: if $\Omega \subset \mathbb{C}$ is a domain, then $\Omega^c$
denotes the complement of $\Omega$ in the extended complex plane $\mathbb{S}$.
\end{rem}


\begin{defi} Let $\Omega$ be a simple connected domain; we define the space $Exp(\Omega)$ as the set of all entire
functions $\phi$ of exponential type such that its Borel transform $B(\phi)$ has all its singularities in $\Omega$ and
such that $B(\phi)$ admits an analytic continuation to $\Omega^c$. This continuation will continue being denoted by
$\mathcal{B}(\phi)$.
\end{defi}

\begin{rem}
Since $\Omega^c$ is closed, the fact that $\mathcal{B}(\phi)$ is analytic in $\Omega^c$ means that there exists an open
set $U\subset \mathbb{S}$ such that $\mathcal{B}(\phi)$ is analytic in $U$ and $\Omega^c\subset U$. Therefore, using the
alternative definition of Borel transform {\rm (\ref{EqCBT})}, we understand $\mathcal{B}(\phi)$ as the  analytic
continuation of its real Laplace transform {\rm (\ref{EqRBT})}.
\end{rem}

\begin{rem}
In what follows, {\bf the Borel transform of} $\phi \in Exp(\Omega)$ always refers to the complex function $\mathcal{B}
(\phi)$ together with an open set $U$ in the extended plane $\mathbb{S}$ such that $\mathcal{B}(\phi)$ is analytic in
$U$ and $\Omega^c\subset U$.
\end{rem}


\begin{defi}
For a function $\phi \in Exp(\Omega)$, we define the set  $H_1(\phi)$ to be the class of closed rectifiable and simple
curves in $\mathbb{C}$ which are pairwise homologous 
and contain the set $s(B(\phi))$ in their bounded regions.
\end{defi}

The following theorem is a classical result about the representation of entire functions of exponential type.

\begin{theorem}\label{Pol_Rep}(Polya's Representation Theorem). Let $\phi$ be a function of exponential type and let $\gamma\in H_1(\phi)$. Then,
$$\phi(z)=\dfrac{1}{2\pi i}\int_{\gamma}e^{sz}\mathcal{B}(\phi)(s)ds.$$
In particular, if $\phi$ is of type  $\tau$  and $R> \tau$, then
$$\phi(z)=\dfrac{1}{2\pi i}\int_{|s|=R}e^{sz}B(\phi)(s)ds.$$
\end{theorem}

This theorem is discussed for instance in \cite[p. 107]{U} and \cite{CPR}. A proof appears in
\cite[Theorem 5.3.5]{BoasBook}.

\begin{defi}
If $d$ is a distribution with compact support in $\mathbb{C}$, we define the $\mathcal{P}$-transform of $d$ by:
$$\mathcal{P}(d)(z):= \, <e^{sz},d>\; , \quad z\in \mathbb{C}.$$
\end{defi}

The $\mathcal{P}$-transform is called the Fourier-Laplace transform in
\cite{U} and the Fourier-Borel transform in Martineau's classical paper \cite{Mar}.
For the particular case in which $d=\mu$ is a complex measure with compact support, the $\mathcal{P}$-transform is
$$\mathcal{P}(\mu)(z) = \underset{\mathbb{C}}\int e^{sz}d\mu (s)\; , z\in \mathbb{C}.$$

\begin{proposition}\label{ProChar}
Let $\mathcal{O} \subset \mathbb{C}$ be a simply connected domain; if $\mu$ is a complex measure with compact support
contained in  $\mathcal{O}$, then $\mathcal{P}(\mu)\in Exp(\mathcal{O})$. Conversely, given any function
$\phi\in Exp(\mathcal{O})$, there exists a complex measure $\mu_{\phi}$
with compact support in $\mathcal{O}$ and such that $\mathcal{P}(\mu_{\phi})(z)=\phi(z)$. The measure $\mu_{\phi}$ is not
unique: it can be chosen to have support on any given curve $\gamma\in H_1(\phi)$.
\end{proposition}
\begin{proof}
Let $K$ be the support of the complex measure $\mu$ (which is of finite variation). The $\mathcal{P}$-transform of
$\mu$ is
$$\mathcal{P}(\mu)(z)=\underset{\mathbb{C}}\int e^{sz}d\mu (s)\; ,$$
which is an entire function. Now, if $R=sup_{s\in K}|s|$, we have
$$|\mathcal{P}(\mu)(z)|\leq \underset{\mathbb{C}}\int e^{R|z|}|d\mu (s)|\leq e^{R|z|}||\mu||\; ,$$
that is, $\mathcal{P}(\mu)$ is an entire function of exponential type.

It remains to show that $s(B(\mathcal{P}(\mu)))\subset \mathcal{O}$. To do that, we compute the Borel transform of
$\mathcal{P}(\mu)$ as the analytic
continuation of its real Laplace transform. Let $z$ be a real number such that $z> R$. Then, we have
\begin{eqnarray*}
B(\mathcal{P}(\mu))(z)&=&\int_0^{+\infty}e^{-zt}\mathcal{P}(\mu)(t)dt\\
&=&\int_0^{+\infty}e^{-zt}\underset{K}\int e^{st}d\mu (s)dt\\
&=&\underset{K}\int \int_0^{+\infty} e^{(s-z)t}dt d\mu (s)\\
&=&\underset{K}\int \dfrac{1}{z-s}d\mu (s)\; ,
\end{eqnarray*}
in which we have used Fubini's theorem. From these computations we have that $\mathcal{P}(\mu) \in Exp(\mathcal{O})$.
In fact, the last integral is the analytic continuation $\mathcal{B}(\mathcal{P}(\mu))$.

To prove the converse implication, let $\gamma\in H_1(\phi)$. Then, Polya's representation theorem (Theorem \ref{Pol_Rep}) means that
$\phi$ can be represented as $\phi(z)=\mathcal{P}(\mu_{\phi})(z)$ for the complex measure $\mu_{\phi}$ defined by

\begin{equation}\label{defMea}
d\mu_{\phi}(s):=\mathcal{B}(\phi)(s)\dfrac{ds}{2\pi i}\; , \quad s\in \gamma\; .
\end{equation}
\end{proof}

\begin{rem}
The analogous of this proposition for general distributions can be found in {\rm \cite{CPR}}. We prefer our version with complex measures
because in this work we do not use the machinary of distributions.
\end{rem}

\subsection{Functions of $\partial_t$ via Borel transform}

Let $\Omega$  
 be a simply connected domain (equivalently, let $\Omega$ be a Runge domain, see \cite[Prop. 17.2]{TrevesBook}). In what follows we denote by $Hol(\Omega)$ the set of holomorphic functions on $\Omega$.

\begin{defi}\label{Def_Oper}
Let $f \in Hol(\Omega)$, $\phi \in Exp(\Omega)$, and assume that $\mu_{\phi}$ is the complex measure defined in {\rm \ref{defMea}} with
compact support on a curve $\gamma \in H_1(\phi)$ so that $\mathcal{P}(\mu_{\phi})=\phi$. We define the operator $f(\partial_t)\phi$
as
$$f(\partial_t)\phi:=\mathcal{P}(f \mu_{\phi})\; .$$
\end{defi}
In Definition \ref{Def_Oper} we assume that the curve $\gamma\in H_1(\phi)$, which defines the measure $\mu_\phi$, is contained in the region $\Omega$. By Cauchy's Theorem, the operator $f(\partial_t)$ is independent of such a $\gamma$ and therefore is well defined.

In this way, using the new measure $f\mu_{\phi}$ and the definition of the $\mathcal{P}$-transform, we see that Equation (\ref{Gen_Eq}) is
understood as the following integral equation

\begin{equation}\label{EqInt}
\int_{\gamma}e^{st}f(s)\mathcal{B}(\phi)(s)\frac{ds}{2\pi i}=g(t),\; \; \gamma \in H_1(\phi).
\end{equation}
We may wonder whether it is necessary to restrict ourselves to Runge domains. Indeed, the counterexample below is proposed in
\cite[pag. 27]{DubinskyBook} in order to show that $f(\partial_t)\phi$ can be multivalued if $\Omega$ is not a Runge domain.
We reproduce it here for the sake of completeness.

\begin{example} \label{runge}
Set $\Omega = \mathbb{C}\setminus \{0\}$. We consider the symbol
$f(s)=\dfrac{1}{s}$ which is holomorphic in $\Omega$ and the function $\phi_{\lambda}(z)=e^{\lambda z}, \lambda \in \Omega$.
Then the set $s(\mathcal{B}(\phi_{\lambda}))=\{\lambda\}$ and we have, for a closed curve  $\gamma$ in $\Omega$ containing $\lambda$,
two possible values for $f(\partial_t)\phi_{\lambda}(z)$: either
$$f(\partial_t)\phi_{\lambda}(z)=\dfrac{1}{\lambda}(e^{\lambda z}-1) \; ,$$
or
$$f(\partial_t)\phi_{\lambda}(z)=\dfrac{1}{\lambda}e^{\lambda z} \; ,$$
depending on whether $\gamma$ encloses the point $\{0\}$ or not.
\end{example}

\smallskip

The following proposition justifies some of the formal computations appearing in physical papers (see \cite{BK1,BK2} and references therein).
It says that the integral operator $f(\partial_t)$ is {\em locally} ({\em i.e.}, whenever $f$ can be expanded as a power series in an adequated
ball contained in $\Omega$) a differential operator of infinite order.

\begin{proposition}\label{IODEq} 
Let $R>0$ and assume that $B_R(0)\subset \Omega$. Suppose that $\phi\in Exp(\Omega)$ is such that $s(\mathcal{B}(\phi))\subset B_R(0)$ and take
$f \in Hol(\Omega)$ with $f(z)=\sum_{k=0}^{\infty}a_kz^k, \, \, |z| < R$. Then, there exist a measure $\mu_{\phi}$ supported on a curve
$\gamma \in H_1(\phi)$ contained in $B_R(0)$ such that $\phi=\mathcal{P}(\mu_{\phi})$, and moreover
$$f(\partial_t)\phi(t)=\mathcal{P}(f \mu_{\phi})(t)=\lim_{l\to \infty}\sum_{k=0}^{l}a_k(\partial_t^k\phi)(t),$$
uniformly on compact sets.
\end{proposition}
\begin{proof}
We note that, since $s(\mathcal{B}(\phi))$ is a discrete set, there exists a real number $\delta > 0$ such that
$dist(s(\mathcal{B}(\phi)),\{s: |s|=R\})< \delta$. From \cite[Theorem 5.3.12]{BoasBook} we have that
$\tau_\phi= \sup_{\omega \in s(\mathcal{B}(\phi)) } |\omega|$ is the type of $\phi$, then there exist a curve
$\gamma \subset (B_{\tau_\phi}(0))^{c}\cap B_R(0)$ such that $\gamma \in H_1(f)$.
Let $\mu_\phi$ be the measure described by Theorem \ref{ProChar} supported on $\gamma $, then $\phi=\mathcal{P}(\mu_{\phi})$.

Moreover, using the measure $\mu_\phi$, we compute:
\begin{eqnarray*}
\dfrac{d}{dz}\phi(z)&=&\dfrac{d}{dz}\mathcal{P}(\mu_{\phi})(z) = \dfrac{d}{dz} \underset{\gamma}\int e^{sz}d\mu_{\phi}(s)
=\underset{\gamma}\int se^{sz}d\mu_{\phi}(s)
=\mathcal{P}(s\mu_{\phi}).
\end{eqnarray*}
From this, we obtain
$$\sum_{k=0}^{l}a_k\dfrac{d^k}{dz^k}\phi(z)-\mathcal{P}(f \mu_{\phi})(z)=
\mathcal{P}\left(\left\{\sum_{k=0}^{l}a_ks^k-f(s)\right\}\mu_{\phi}\right)(z)\; ,$$
and therefore
\begin{eqnarray*}
\left| \sum_{k=0}^{l}a_k\dfrac{d^k}{dz^k}\phi(z)-\mathcal{P}(f \mu_{\phi,R})(z)\right| & = &
                                              \left| \underset{\gamma}\int e^{sz}\{\sum_{k=0}^{l}a_ks^k-f(s)\}d\mu_{\phi}(s) \right| \\
&\leq& \underset{\gamma}\int e^{|z||s|}\, \left|\sum_{k=0}^{l}a_ks^k-f(s)\right| \, |d \mu_{\phi}|(s) \; .
\end{eqnarray*}
Now we take limits as $l\to \infty$. The result follows from the Lebesgue dominated convergence theorem. We note that the convergence
is uniform over compact subsets of $\mathbb{C}$.
\end{proof}
Thus, under the hypothesis of this proposition, we have seen that Equation (\ref{Gen_Eq}) becomes ``locally" the following infinite order
differential equation:
\begin{equation}\label{EqIOD}
\sum_{k=0}^{\infty}a_k\dfrac{d^k}{dt^k}\phi(t)=g(t).
\end{equation}

Interestingly, Proposition \ref{IODEq} also shows that $f(\partial_t)$ is linear on the space of functions $\phi$ satisfying the hypothesis
appearing therein. We now show that linearity is true in general:

\begin{lemma}
The operator $f(\partial_t): Exp(\Omega) \to Exp(\Omega)$ is linear.
\end{lemma}
\begin{proof}
 For $\phi,\psi\in Exp(\Omega)$ we have that $s(\mathcal{B}((\phi+\psi)))\subseteq s(\mathcal{B}(\phi))\cup s(\mathcal{B}(\psi))\subset \Omega$.
  Let $\gamma\in H_1(\phi+\psi)$ such that $s(\mathcal{B}(\phi))\cup s(\mathcal{B}(\psi))$ is enclosed by $\gamma$.
  This implies that $\gamma\in H_1(\phi)\; , \gamma \in H_1(\psi)$; then
 \begin{eqnarray*}
f(\partial_t)(\phi+\psi)(t)&=&\mathcal{P}(f \mu_{\phi+\psi})(t)=\int_{\gamma}e^{st}f(s)B(\phi+\psi)(s)\dfrac{ds}{2\pi i} \\
 &=&\int_{\gamma}e^{st}f(s)B(\phi)(s)\dfrac{ds}{2\pi i}+\int_{\gamma}e^{st}f(s)B(\psi)(s)\dfrac{ds}{2\pi i}\\
 &=&f(\partial_t)(\phi)(t)+f(\partial_t)(\psi)(t).
 \end{eqnarray*}
\end{proof}

\begin{rem}
We remark that, if $\phi,\psi\in Exp(\Omega)$ then the inclusions $H_1(\phi+\psi)\subseteq H_1(\phi)$ and $H_1(\phi+\psi)\subseteq H_1(\psi)$
do not necessarily hold. This is why we had to choose an appropriated curve $\gamma$ to carry out the above proof.
For example, let $R>2$ and set $\lambda \in \mathbb{C}$ with
$|\lambda|=3/2$; also consider two functions
$g_1,g_2 \in Exp_{1/2}(\mathbb{C})$; then, the functions $f_1(z)=e^{\lambda z}+g_1(z)$ and $f_2(z)=-e^{\lambda z}+g_2(z)$ are elements
in $Exp(B_R(0))$. Furthermore, if $\gamma \in H_1(f_1+f_2)$ with $\gamma \subset B_{1/2}(0)^{c}\cap B_1(0)$, then $\gamma \not  \in H_1(f_1)$
and also $\gamma \not  \in H_1(f_2)$.
\end{rem}

\begin{rem}\label{remOp}
Let us assume that  $Exp(\Omega)$ is endowed with the topology of uniform convergence on compact sets. With this topology, the operator
$f(\partial_t)$ is not bounded: It is enough to take $\Omega=\mathbb{C}$ and as symbol $f$ the identity map. The following specific example shows
less trivially that the linear operator $f(\partial_t)$ is not necessarily continuous:

Let $\Omega=\mathbb{C}\setminus \mathbb{R}^{+}_0$ (a Runge Domain) and consider the symbol $f(s)=\dfrac{1}{s}$. Then $f\in Hol(\Omega)$;
we consider the sequence $\phi_n(z)=e^{\frac{i}{n}z}-e^{\frac{-i}{n}z}$. We have
$$s(\mathcal{B}(\phi_n))=\left\{-\frac{i}{n},\frac{i}{n} \right\}\subset \Omega\; .$$
We can see that $\phi_n\to 0$ in the topology of uniform convergence on compact sets. On the other hand we have (See Example \ref{runge})
$$f(\partial_t)(\phi_n)(z)=\frac{n}{i}\left( e^{\frac{i}{n}z}+e^{\frac{-i}{n}z} \right)\; ,$$
and considering the compact ball $\overline{B_k(0)}$ with centre the origin and radius $k$, we have
$$\sup_{z\in \overline{B_k(0)}}|f(\partial_t)(\phi_n)(z)|\geq |f(\partial_t)(\phi_n)(0)|=2n$$
which goes to infinity when $n\to \infty$.
\end{rem}

The following lemma says that nonlocal equations involving $f(\partial_t)$ can be solved in $Exp(\Omega)$:

\begin{lemma}\label{SurjLem}
The operator $f(\partial_t): Exp(\Omega) \to Exp(\Omega)$ is surjective.
\end{lemma}
\begin{proof}
The surjectivity of the operator comes from the solvability of the following equation
\begin{equation}
f(\partial_t)\phi=g, \,\, \quad g\in Exp(\Omega) \; .
\end{equation}
Since the zet of zeros of $f$, $\mathcal{Z}(f)$ say, is a set of isolated points and
$g\in Exp(\Omega)$, there is a curve $\gamma\in H_1(g)$ such that $\mathcal{Z}(f)\cap \gamma= \emptyset$.
Also, there is a measure $\mu_g$ suported on $\gamma$ such that $g=\mathcal{P}(\mu_g)$. Set $\phi=\mathcal{P} (\dfrac{\mu_g}{f})$; i.e.
$$\phi(z)=\dfrac{1}{2\pi i}\int_{\gamma} \dfrac{ e^{z\eta}\mathcal{B}(g)(\eta)}{f(\eta)}d\eta\; .$$
It is evident that $\phi\in Exp(\mathbb{C})$, now we want to see that $s(\mathcal{B}(\phi))\subset \Omega$. For that,
let us calculate the Borel Transform of $\phi$ as the analytic continuation of its real Laplace transform. Let $z\in \mathbb{R}$ be
sufficiently large so that $|Re(\eta)|<z\; $ for all $\eta \in \gamma$; we have
\begin{eqnarray*}
\int_0^{+\infty}e^{-zt}\phi(t)dt &=& \int_0^{+\infty}e^{-zt}\dfrac{1}{2\pi i}\int_{\gamma} \dfrac{ e^{z\eta}\mathcal{B}(g)(\eta)}{f(\eta)}d\eta dt\\
&=&\dfrac{1}{2\pi i}\int_{\gamma} \int_0^{+\infty}\dfrac{ e^{t(\eta-z)}\mathcal{B}(g)(\eta)}{f(\eta)}d\eta\\
&=&\dfrac{1}{2\pi i}\int_{\gamma}\dfrac{\mathcal{B}(g)(\eta)}{f(\eta)(z-\eta)}d\eta\; ,
\end{eqnarray*}
in which we have used Fubini's Theorem. Therefore,
$$\mathcal{B}(\phi)(z)=\dfrac{1}{2\pi i}\int_{\gamma}\dfrac{\mathcal{B}(g)(\eta)}{f(\eta)(z-\eta)}d\eta\; ,$$
and using Morera's theorem, we can see that $\mathcal{B}(\phi)$ is analytic in $\Omega^{c}$; thus $s(\mathcal{B}(\phi))\subset \Omega$.
On the other hand, it is not difficult to see that it satisfies $f(\partial_t)\phi=g$.
\end{proof}


\subsection{A representation formula for solutions to $f(\partial_t)\phi = g$}

\begin{proposition}\label{FDim}

Let $f \in Hol(\Omega)$ and denote $\mathcal{Z}(f)$ for the set of its zeros. A function $\phi \in Exp(\Omega)$
 of exponential type $\tau_{\phi}$ is solution to the homogeneous equation $f(\partial_t)\phi=0$ if and only if there exist polynomials
 $p_k$ of degree less than the multiplicity of the root $s_k \in \mathcal{Z}(f) \cap B_{\tau_{\phi}}(0)$, such that
$$\phi(t)=\sum_{\substack{s_k\in \mathcal{Z}(f) \\ |s_k|<\tau_{\phi}}} p_k(t) e^{t s_k}\; .$$
\end{proposition}
\begin{proof}({\sl Sufficiency}) in order to prove that the function
$$\phi(t)=\sum_{\substack{s_k\in \mathcal{Z}(f)\\ |s_k|<\tau_{\phi}}} p_k(t) e^{t s_k}\; ,$$
is solution to the homogeneous equation $f(\partial_t)\phi=0$, it is enough to see that for a given $k$ and $s_k\in \mathcal{Z}(f)$
with $|s_k|<\tau_{\phi}$, the following holds:
$$f(\partial_t)(p_k(t)e^{ts_k})=0.$$

Indeed, we first note that for a natural number $d$ and a complex number $a_d$, the Borel transform of $a_d s^de^{\lambda s}$ is
\begin{eqnarray}\label{NewIdent}
\mathcal{B}(a_d s^de^{\lambda s})= a_d\frac{d!}{(s-\lambda)^{d+1}}\; .
\end{eqnarray}
Now, let $s_k$ be a zero of $f$ of order $d_k+1$, $p_k$ a polynomial of degree $deg(p_k)\leq d_k$ and suppose that
$\gamma_k \in H_1(p_k(z)e^{zs_k})$; then, using linearity of the Borel transform and the Cauchy theorem we have
$$f(\partial_t)\left(p_k(t)e^{ts_k} \right)=\dfrac{1}{2\pi i}\int_{\gamma_k}e^{t\eta}f(\eta)\mathcal{B}(p_k(z)e^{zs_k})(\eta) d\eta =0\; .$$
From these computations, we deduce that
$$f(\partial_t)\left( \sum_{s_k\in \mathcal{Z}(f): |s_k|<\tau_{\phi}} p_k(t)e^{ts_k} \right)=0.$$
On the other hand, it is evident that
$$\phi(t)=\sum_{\substack{s_k\in \mathcal{Z}(f) \\ |s_k|<\tau_{\phi}}} p_k(t)e^{t s_k}\; ,$$
has exponential type $\tau_{\phi}$ and  from (\ref{NewIdent}) we conclude that $\phi \in  Exp(\Omega)$ .

Before proving {\sl necessity}, we must ensure finite dimensionality of a vector space to be defined below. We write this fact as a separate
result, because it is interesting in its own right; after that we will finish the proof of this proposition.
\end{proof}

We use the following notations. Let $\mathcal{R}$ be the closure of a bounded and simply connected region which does not contain any singularity
of $f$ and let $\gamma$ denotes its boundary. We also denote by $A(\mathcal{R})$ the set of continuous functions that are analytic in the interior
of $\mathcal{R}$ endowed with the supremum norm and, for $z \in \mathbb{C}$, we let $E_z:\mathcal{R}\to\mathbb{C}$ be the complex  exponential function $E_z(\xi)=e^{z\xi}$. Finally, we let
$$\mathcal{M}_{f,\gamma}:=cl(span\{E_z\cdot f: z\in \mathbb{C}\})\; ,$$
where $cl$ denotes closure in $A(\mathcal{R})$.

\begin{lemma}
Let $\{s_k\}_{k=1}^K$ be an enumeration of all the zeros of $f$ in $\mathcal{R}$ and let $m_k$ denote their corresponding multiplicities. Then
\begin{equation} \label{zeroes}
\mathcal{M}_{f,\gamma}=
\{\psi\in A(\mathcal{R}): \psi \, \, is \, \, zero \, \, at\, \, s_k\, \, with \, \, multiplicity \, \, \geq m_k, \, \, 1\leq k\leq K\}\; .
\end{equation}
\end{lemma}
\begin{proof}
It is not difficult to see that $\mathcal{M}_{f,\gamma}$ is a subset of the set appearing in the right hand side of (\ref{zeroes}).
Let us prove the other inclusion. If $\psi$ belongs to the right
hand side of (\ref{zeroes}), then $\dfrac{\psi}{f}\in A(\mathcal{R})$. Since $\mathcal{R}$ is compact with connected complement, by
Mergelyan's Theorem  (see \cite[Theorem 20.5]{WRudin} and \cite{SMerge1,SMerge2}) we know that the set of polynomials $Pol$ is dense
in $A(\mathcal{R})$. Therefore, given $\epsilon >0$ there is a polynomial $p\in \, Pol$ such that
$||\dfrac{\psi}{f}-p||_{A(\mathcal{R})}<\epsilon$. It follows that
$\psi \in cl(f \cdot Pol)$. Now we note that
$$Pol\subset cl(span\{E_z:z\in \mathbb{C}\}).$$
Indeed, it is sufficient to note that the right hand side is an algebra which contains $1$ and $\xi$ for any $\xi \in \mathcal{R}$,
and certainly, we have
$\xi=\lim\limits_{n\to\infty}\dfrac{e^{\xi 1/n}-1}{1/n}$. Therefore
$$\psi\in cl(f \cdot Pol)\subset cl(span\{f \cdot E_z: z\in \mathbb{C}\})\; .$$
\end{proof}

A special case of this lemma is in \cite[Lemma 5.4]{CPR}.

\begin{lemma}\label{Dimension}
Under the conditions of previous lemma, the space $A(\mathcal{R})/M_{f,\gamma}$ has dimension $M=m_1+m_2+\cdots +m_K$.
\end{lemma}
\begin{proof}
We can note that $M_{f,\gamma}=(\prod_{k=1}^K (z-s_k)^{m_k})$ is a closed ideal of $A(\mathcal{R})$.
First of all, given a complex number $\omega \in \mathcal{R}$ and an integer number $m>0$ we claim that the quotient space
$A(\mathcal{R})/((z-\omega)^{m})$ has dimension $m$ and that a basis is given by the set
$\beta_1:=\{\overline{1},\overline{z-\omega},\overline{(z-\omega)^2},\cdots,\overline{(z-\omega)^{m-1}}\}$ (here an overline indicates
equivalence class). In fact, let $\alpha_0,\alpha_1,\cdots ,\alpha_{m-1}$ be complex numbers, then
$$\sum_{l=0}^{m-1}\alpha_l(z-\omega)^{l}=0\, \, \quad \mbox{ belongs to } \quad \, \, A(\mathcal{R})/((z-\omega)^{m})$$
if and only if $\sum_{l=0}^{m-1}\alpha_l(z-\omega)^{l} \in ((z-\omega)^{m})$. Thus, there exists a function $\psi \in ((z-\omega)^{m})$ such that
$\sum_{l=0}^{m-1}\alpha_l(z-\omega)^{l}=\psi(z)$, and therefore $\alpha_l=\dfrac{1}{l!}\dfrac{d^l}{dz^l}\psi(z)|_{z=\omega}=0,$ for
$\, \, 0\leq l \leq m-1$. It follows that $\beta_1$ is a linearly independent set. Now let $\overline{h}\in A(\mathcal{R})/((z-\omega)^{m})$
and consider the complex numbers $\alpha_l=\dfrac{1}{l!}\dfrac{d^l}{dz^l}h(z)|_{z=\omega}, \, \, l=0,1,\cdots, m-1$; then
$$h(z)=\sum_{l=0}^{m-1}\alpha_l(z-s_1)^{l}\, \, \quad \mbox{ belongs to } \quad \, \, A(\mathcal{R})/((z-\omega)^{m})\; .$$

\noindent As a second step, we show that for any $k\in \{1,2,\cdots,K\}$ the following equality holds
$$A(\mathcal{R})/M_{f,\gamma}=A(\mathcal{R})/((z-s_k)^{\sum_{j=1}^K m_j})\; .$$
In fact, fix any $k\in\{1,2,\cdots,K\}$ and set $h\in A(\mathcal{R})/M_{f,\gamma},$ then
$$h(z)=r(z)+\prod_{i=1}^K (z-s_i)^{m_i}\psi_0(z)\; .$$
Also, we have
$$(z-s_i)^{m_i}=((z-s_k)+(s_k-s_i))^{m_i}=\sum_{n=0}^{m_i} a_n(z-s_k)^{m_i-n}(s_k-s_i)^{n}=(z-s_k)^{m_i} + p_i(z)\; ,$$
where the polynomial $p_i$ has degree $m_i-1$; therefore we obtain  
$$h(z)=r(z)+\prod_{i=1}^K (z-s_i)^{m_i}\psi_0(z)= r(z)+r_1(z)+(z-s_k)^{\sum_{j=1}^K m_j}\psi_0(z),$$
which implies that $h\in A(\mathcal{R})/((z-s_k)^{\sum_{j=1}^K m_j})$.\\

\noindent Conversely, let $h\in A(\mathcal{R})/((z-s_k)^{\sum_{j=1}^K m_j})$. Then
$$h(z)= r(z)+(z-s_k)^{\sum_{j=1}^K m_j}\psi_0(z)=r(z)+\prod_{i=1}^K (z-s_k)^{m_i}\psi_0(z)\; .$$
Now, as in the previous step, we have
$$(z-s_k)^{m_i}=((z-s_i)+(s_i-s_k))^{m_i}=\sum_{n=0}^{m_i} a_n(z-s_i)^{m_i-n}(s_i-s_k)^{n}=(z-s_i)^{m_i} + p_i(z),$$ where the polynomial $p_i$
has degree $m_i-1$. It follows that
$$h(z)=r(z)+r_1(z)+\prod_{i=1}^K (z-s_i)^{m_i}\psi_0(z)\; ,$$
which implies that $h\in A(\mathcal{R})/M_{f,\gamma}\,$. Therefore, using the first step, we conclude that $m_1+m_2+\cdots+m_K$ is
the dimension of the quotient space $A(\mathcal{R})/M_{f,\gamma}\,$, as claimed.
%
%
%
\end{proof}

\noindent As an immediate consequence of this lemma we have:

\begin{cor}
 Let $\Omega \subset \mathbb{C}$ be an unbounded domain and assume that $f\in Hol(\Omega)$ has an infinite number of zeros $\{s_k\}$
 with multiplicity  $m_k$  for $k \in \{1,2,3, \cdots\}$. Then the space $A(\Omega)/M_{f}$ has infinite dimension, where
\begin{equation*}
\mathcal{M}_{f}=
\{\psi\in Hol(\Omega): \psi \, \, is \, \, zero \, \, at\, \, s_k\, \, with \, \, multiplicity \, \, \geq m_k, \, \, 1\leq k < \infty\}\; .
\end{equation*}
%
\end{cor}

\noindent {\sl Now we proceed to finish the proof of Proposition \ref{FDim}.}
\begin{proof}
Let $\phi \in Exp(\Omega)$ be given. From \cite[Theorem 5.3.12]{BoasBook} we have that its type is
$\tau_{\phi}=\max_{\omega \in s(\mathcal{B}(\phi))}|\omega|$. Since $s(\mathcal{B}(\phi)) \subset \Omega$ and is a discrete set, we can find a
curve $\gamma$ in $\Omega$ whose enclosed region $\mathcal{R}$ contains the set $s(\mathcal{B}(\phi))$ (i.e $\gamma \in H_1(\phi)$) and
such that it also contains all zeros $s_i \in \Omega$ of the symbol $f$ with $ |s_i|<\tau_{\phi}\,$. Let $\{s_i\}_{i=1}^k$ be an enumeration of
the zeros of $f$ in $\mathcal{R}\cap B_{\tau_{\phi}}(0)$  and let $m_i$ denote their corresponding multiplicities. We note also that (using
Proposition \ref{ProChar}) we know that there exist a measure $\mu$ supported on $\gamma \in H_1(\phi)$ such that $\phi=\mathcal{P}(\mu)$.

Now, we note that an element $h\in A(\mathcal{R})/M_{f,\gamma}$ is completely determined by the following set
\begin{equation}\label{Elem}
\left\{\dfrac{d^j}{dz^j}h(z)|_{z=s_i}: 0\leq j\leq m_i-1;\, \, 1\leq i\leq k\right\}\; .
\end{equation}
From Lemma \ref{Dimension}, we have that $A(\mathcal{R})/M_{f,\gamma}$ has dimension $m_1+m_2+\cdots +m_k$; therefore its dual space has the
same dimension. Moreover, it is not difficult to see that the following collection of linear functionals
$$\left\{d_{i,j}=\dfrac{d^j}{dz^j}|_{z=s_i}: 0\leq j\leq m_k-1;\, \, 1\leq i\leq k \right\}\; ,$$
are $m_1+m_2+\cdots m_k$-elements in the space $ (A(\mathcal{R}))^*$ which annihilate $M_{f,\gamma}$; therefore they induces the following
$m_1+m_2+\cdots m_k$-linearly independent elements in the dual space of $ A(\mathcal{R})/M_{f,\gamma}$
$$  \{\widetilde{d_{i,j}}: 0\leq j\leq m_k-1;\, \, 1\leq i\leq k\}\; ;$$
where $\widetilde{d_{i,j}}(\overline{\phi})=d_{i,j}(\phi)$ for $\overline{\phi}\in A(\mathcal{R})/M_{f,\gamma}\,$.
Consequently, every element $\varrho \in (A(\mathcal{R})/M_{f,\gamma})^*$ can be written in the form
$$\varrho =\sum_{i=1}^k \sum_{j=0}^{m_i-1}a_{i,j}\widetilde{d_{i,j}} $$
for some $a_{i,j}\in \mathbb{C}$. Now, given and element $\phi \in  A(\mathcal{R})$, there exist a unique $r\in A(\mathcal{R})$ such that
$\overline{\phi}= \overline{r}$ in $A(\mathcal{R})/M_{f,\gamma}$, and using the characterization given in (\ref{Elem}) we have
$$\varrho(\overline{\phi})=\varrho(\overline{r})=\sum_{i=1}^k \sum_{j=0}^{m_i-1}a_{i,j}d_{i,j}(r)=
\sum_{i=1}^k \sum_{j=0}^{m_i-1}a_{i,j}\dfrac{d^j}{dz^j}(\phi)|_{z=s_i}\; .$$
On the other hand, from the equation $\mathcal{P}(f\cdot\mu)=0$ we have that the measure $\mu$ defines a functional on $A(\mathcal{R})$ which
annihilates $M_{f,\gamma}$ and it induces a functional $\widetilde{\mu}$ on $A(\mathcal{R})/M_{f,\gamma}\,$. Therefore, there exist complex
numbers $b_{i,j}$ such that
$$\widetilde{\mu}=\sum_{i=1}^k \sum_{j=0}^{m_i-1}b_{i,j}\widetilde{d_{i,j}}\; .$$
 Then, we have
\begin{eqnarray*}
\phi(t)&=&\mathcal{P}(\mu)(t)=
\int_{\gamma}e^{tz}d\mu(z)=\widetilde{\mu}(\overline{e^{tz}})=\sum_{i=1}^k \sum_{j=0}^{m_i-1}a_{i,j}\dfrac{d^j}{dz^j}(e^{tz})|_{z=s_i}\\
&=&\sum_{i=1}^k \left( \sum_{j=0}^{m_i-1}a_{i,j}t^j \right) e^{ts_i}\\
&=&\sum_{i=1}^k p_i(t)e^{ts_i}\; .
\end{eqnarray*}
The proof of Proposition \ref{FDim} is finished.
\end{proof}


\begin{cor}\label{corcor}
Let $R>0$ and $f\in Hol(B_R(0))$. Then, a function $\phi \in Exp(B_R(0))$ of exponential type $\tau_{\phi}$ is a solution of the homogeneous equation  $f(\partial_t)\phi=0$, if and only if there exist polynomials $p_k$ of degree less than the multiplicity of the root $s_k \in \mathcal{Z}(f)\cap B_{\tau_{\phi}}(0)$, such that
$$\phi(t)=\sum_{\substack{s_k\in \mathcal{Z}(f) \\ |s_k|<\tau_{\phi}}} p_k(t)e^{ts_k}\; .$$
\end{cor}
In particular, if the symbol $f$ is an entire function, we deduce Theorem  5.1 in \cite{CPR} from Corollary \ref{corcor}. The following theorem
is an easy application of the previous results; it generalizes Proposition \ref{FDim}.
\begin{theorem}\label{TheSol}
Let $f \in Hol(\Omega)$ and
$g\in Exp(\Omega)$. Then a function $\phi \in Exp(\Omega)$ of exponential type $\tau_{\phi}$ is solution for the non-homogeneous equation $f(\partial_t)\phi=g$ if and only if there exist polynomials $p_k$ of degree less than the multiplicity of the root $s_k \in \mathcal{Z}(f)\cap B_{\tau_{\phi}}(0)$, such that
$$\phi(t)=\mathcal{P}\left( \dfrac{\mu_{g}}{f} \right)(t)+\sum_{\substack{s_k\in \mathcal{Z}(f) \\ |s_k|<\tau_{\phi}}} p_k(t)e^{ts_k}\; .$$
\end{theorem}


\section{Linear zeta-nonlocal field equations}

Now we apply the theory developed in the previous section to find explicit solutions of the following linear
zeta-nonlocal field equation:

\begin{equation}\label{EquEnt}
\zeta(\partial^2_t+h)\phi=g\;  ,
\end{equation}
in which $h$ is a real parameter. Our solution depends crucially on the properties of $g$. We show that if $g$ is of exponential type, then so is
$\phi$ and solving (\ref{EquEnt}) explicitly is rather straightforward. However, if the data $g$ is not of exponential type, analysis become
very delicated. We consider this problem in 4.2, in which we assume that the Laplace transform $\mathcal{L}(g)$ exists and it has an analytic
extension to an appropriated angular contour. In section we use
notation introduced in Section 2.

\subsection{Zeta-nonlocal field equation with source function in $Exp(\Omega)$}

Equation (\ref{EquEnt}) can be solved completely in the space of entire functions of exponential type. Since Equation (\ref{EquEnt}) depends on
the values of $h$, we study it in three different cases:

\subsubsection{ {\bf Case $h>1$.}}
In this case the symbol $\zeta_h\circ p(s)=\zeta(s^2+h)$ has poles $i\sqrt{h-1}$ and $-i\sqrt{h-1}$. As we have already pointed out,
the behavior of $\zeta_h\circ p(s)$ can be represented in the following picture:
\begin{center}
\includegraphics[width=4cm, height=4cm]{111}
\end{center}
{\footnotesize
The poles of $\zeta(s^2+h)$ are the vertices of dark hyperbola, indicated by two thick dots. The trivial zeros of
$\zeta(s^2+h)$ are indicated by thin dots on the imaginary axis; and the non-trivial zeros are located on the darker
painted region (critical region).}
\\

Now, let us consider the simply connected domain
$$\Omega := \mathbb{C} \setminus \left\{ s\in \mathbb{C}: Re(s)\geq 0, \; \; |Im(s)|=\sqrt{h-1} \right\}\; .$$
We see that the symbol $\zeta_h(s)$ is holomorphic in $\Omega$, and therefore for a source function $g \in Exp(\Omega)$, equation (\ref{EquEnt})
is the following integral equation for the measure $\mu_{\phi}$:
\begin{equation}\label{ZIntEq}
\mathcal{P}((\zeta_h\circ p)\cdot \mu_{\phi})(t)=g(t)\; .
\end{equation}
\begin{theorem}\label{TeoIEZ}
Let $g\in Exp(\Omega)$. 
Then a function $\phi \in Exp(\Omega)$ of exponential type $\tau_{\phi}$ is solution for the integral equation $(\ref{ZIntEq})$
if and only if there exist polynomials $p_k$ of degree less than the multiplicity of the root $s_k \in \mathcal{Z}(\zeta_h\circ p)\cap B_{\tau_{\phi}}(0)$, such that
$$\phi(t)=\int_{\gamma}\dfrac{e^{ts}}{\zeta(s^2+h)} d\mu_g(s)+ \sum_{\substack{s_k\in \mathcal{Z}(\zeta_h\circ p) \\ |s_k|<\tau_{\phi}}} p_k(t)e^{ts_k}.$$
Where $\gamma \in H_1(g)$ and enclose the root $s_k \in \mathcal{Z}(\zeta_h\circ p)\cap B_{\tau_{\phi}}(0)$.
\end{theorem}

\begin{rem}
In this theorem (and also in the results that follow) we find that
the solution $\varphi(t)$ depends on polynomials $p_k(t)$. These
polynomials are calculated using the zeroes (and their orders) of the function $\zeta_h\circ p$, see Proposition 3.6 and Theorem 3.11. We comment further on this in Subsection 4.2.
\end{rem}

On the other hand, we can note that for given  $R<\sqrt{h-1}$, the domain $\Omega$ contains the ball $B_R(0)$, and since the symbol $\zeta_h\circ p(s)$ is analytic in this ball, it can be expressed there in its Taylor series representation, say
$$\zeta(s^2+h)= \sum_{k=0}^{\infty}a_k(h)s^k\; , |s|<R\; .$$
Therefore, using proposition \ref{IODEq}, in the space $Exp(B_R(0))$ we have that equation (\ref{EquEnt}) can be viewed as the following infinite order differential equation
\begin{equation}\label{ZIODEq}
\sum_{k=0}^{\infty}a_k(h)\dfrac{d^k}{dt^k}\phi(t)=g(t)\; .
\end{equation}
In this situation, we have the following result:
\begin{theorem}\label{TeoIOZ}
Let $R<\sqrt{h-1}$ and $g\in Exp(B_R(0))$.
 Then, a function $\phi \in Exp(B_R(0))$ of exponential type $\tau_{\phi}$ is solution of the infinite order zeta-nonlocal field equation $(\ref{ZIODEq})$
if and only if there exist polynomials $p_k$ of degree less than the multiplicity of the root $s_k \in \mathcal{Z}(\zeta_h\circ p) \cap B_R(0)$, such that 
 $$\phi(t)=\int_{|s|=R}\dfrac{e^{ts}}{\zeta(s^2+h)} d\mu_g(s) +\sum_{\substack{s_k\in \mathcal{Z}(\zeta_h\circ p) \\ |s_k|<\tau_{\phi}}} p_k(t)e^{ts_k}\; . $$
\end{theorem}

\subsubsection{ {\bf  Case $h<1$.}}
In this case we have that $\sqrt{1-h}$ and $-\sqrt{1-h}$ are the poles of the symbol $\zeta_h\circ p$. The behavior of $\zeta_h\circ p$ is represented in the following picture
\begin{center}
\includegraphics[width=4cm, height=4cm]{222}
\end{center}
{\footnotesize
The poles of $\zeta(s^2+h)$ are the vertices of dark hyperbola, indicated by two thick dots. The trivial zeros of
$\zeta(s^2+h)$ are indicated by thin dots on the real axis; the non-trivial zeros are located on the darker painted
region (critical region).}

Therefore choosing as our basic region $\Omega$ the following domain:
$$\Omega := \mathbb{C} \setminus \left\{ s\in \mathbb{C}: Im(s)\geq 0, \; \; |Re(s)|=\sqrt{1-h} \right\}\; ,$$
we can obtain theorems for the equation $\zeta(\partial_t^2+h)\phi=g$ which are similar to Theorem \ref{TeoIEZ} and Theorem \ref{TeoIOZ}.

\subsubsection{ {\bf Case $h=1$.}}

In this case there is a pole at $s=0$. We saw that the behavior of $\zeta_1\circ p$ is represented in the following picture
\begin{center}
\includegraphics[width=4cm, height=4cm]{333}
\end{center}
{\footnotesize The pole of $\zeta(s^2+1)$ is the origin (vertex of
dark curves $y=|x|, y=-|x|$). The trivial zeros of $\zeta(s^2+h)$
are indicated by thin dots on the imaginary axis; the non-trivial
zeros are located on the darker painted region (critical
region).}
\\

Let $\Omega$, be the region
$$\Omega := \mathbb{C} \setminus \{s\in \mathbb{C}: Re(s)\geq 0, \; \; |Im(s)|=0 \} \; .$$
Since we cannot construct a ball around the origen in which $\zeta_1 \circ p$ is analytic, we cannot obtain a result
analogous to Theorem \ref{TeoIOZ} for the equation
$$\zeta(\partial_t^2)\phi=g\; , $$
this is, we {\em do not have} an ``infinite order equation" but a genuine nonlocal equation.
On the other hand, it is possible to state a result analogous to Theorem \ref{TeoIEZ}. We omit details.

\subsection{Zeta-nonlocal field equation with source function in $\mathcal{L}_{>}(\mathbb{R}_+)$}

Now we consider the case in which the source function $g(t), t\geq 0$ is an analytic function not necessarily of exponential
type. We assume that it possesses Laplace transform, and therefore there exists a real number $a$ such that the following
integral
$$\mathcal{L}(g)(z)=\int_0^{\infty}e^{-tz}g(t)dt \; ,$$
converges absolutely and uniformly on the half-plane $\{z\in \mathbb{C} : |z|>a\}$, and for which the function
$z\to \mathcal{L}(g)(z)$ is analytic. We also assume that $\mathcal{L}(g)$ has an analytic extension to the left of
$Re(s)=a$ until a singularity $a_0$, and that this new region of
analyticity has an angular contour $\kappa_{\infty}$ as its boundary. 

Hereafter we denote by $\mathcal{L}_{>}(\mathbb{R}_+)$ the space of analytic functions that possess the properties
described above.

The problem of interest in this situation is to solve the following equation
\begin{equation}\label{ExEq}
\zeta(\partial_t^2+h)f=g,
\end{equation}
for a given $g\in \mathcal{L}_{>}(\mathbb{R}^+)$, where the operator $\zeta(\partial_t^2+h)$ needs to be properly defined
in order to have a correctly posed problem. The solution of Equation (\ref{ExEq}) if it exists, will not necessarily be an
entire function of exponential type.

Let $g\in \mathcal{L}_{>}(\mathbb{R}_+)$ and let the first singularity of the analytic extension of  $\mathcal{L}(g)$ up to
an angular contour $\kappa_{\infty}$ be $a_0=0$. Now consider an angle $\frac{\pi}{2} < \psi \leq \pi$, a positive real number
$r>0$ and let $\kappa_r$ be a finite angular contour contained in $\kappa_{\infty}$. Concretely, $\kappa_r$ is composed by a
circular sector of radius $\delta$ centered at the origen and the respective rays of opening $\pm \psi$ as in the following
picture:
\begin{center}
\includegraphics[width=6cm, height=5cm]{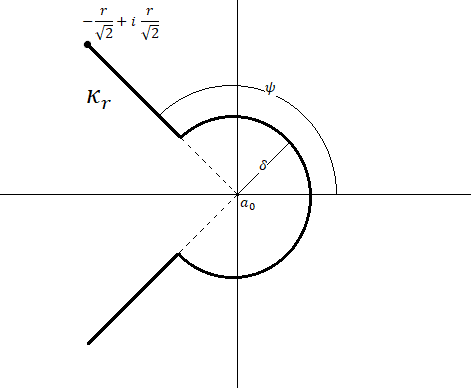}
\end{center}

\noindent Now, let us pick the complex measure
$$d\mu_r(s):=\mathcal{X}_{\kappa_r}(s)\mathcal{L}(g)(s)\frac{ds}{2\pi i}\; ,$$
where $\mathcal{X}_{\kappa_r}$ denotes the characteristic function of the contour $\kappa_r$. This measure allows us to
define the following function $g_r:\mathbb{C}\to \mathbb{C}$ using $\mathcal{P}$-transform:
$$g_r(z):=\mathcal{P}(\mu_r)(z)=\int_{\kappa_r}e^{zs}\mathcal{L}(g)(s)\frac{ds}{2\pi i}\; .$$
\begin{lemma}\label{LemFin1} We have:
\begin{enumerate}
 \item The function $g_r$ is an entire function of order $1$ and exponential type $r$. 
 \item For each $r>0$, the analytic continuation of the Borel Transform of $g_r$ is $\mathcal{B}(g_r)(z)=K*\mu_r(z)$,
 where $K(z)= 1/z$, and its conjugate diagram is the convex hull of the contour $\kappa_r$. In particular, if we consider
 the measure
 $$d\mu_{g_r}(s)=K*\mu_r(s)\frac{ds}{2\pi i}\; ,$$ then $g_r=\mathcal{P}(\mu_r)=\mathcal{P}(\mu_{g_r}).$
\end{enumerate}
\end{lemma}

\begin{proof}
\noindent We prove item {\sl 1}. Let $r>0$ fixed; %
first, we note that for $n\geq 0$
$$g^{(n)}_r(0)=\int_{\kappa_r}s^n\mathcal{L}(g)(s)\dfrac{ds}{2\pi i}\; .$$
Now, defining 
$$M_r:=\dfrac{1}{2\pi}\int_{\kappa_r}|\mathcal{L}(g)(s)ds|\; ,$$
we obtain
$$\dfrac{\ln{|g^{(n)}_r(0)|}}{n\ln n}\leq \dfrac{\ln{r^nM_r}}{n\ln n} = \dfrac{n\ln{r}+\ln{M_r}}{n\ln n}\; ,$$
which approaches zero as $n\to \infty$. Therefore using formula (\ref{Order}) we obtain the order of $g_r$ as
$$\rho = \left(1-\lim_{n\to \infty}\sup\dfrac{\ln{|g_{r}^{(n)}(0)|}}{n\ln(n)}\right)^{-1}=1$$
With this information, we compute the type of $g_r$ using formula (\ref{Type}), 
$$\sigma=\lim_{n\to \infty}\sup |g^{(n)}_r(0)|^{1/n}.$$
It is not difficult to see that $\sigma \leq r$; we will conclude that $\sigma =r$ by considering the region of analyticity
of the Borel transform of $g_r$ using item {\sl 2}.\\

\noindent {\sl 2}. Since $\kappa_r$ is compact we have
$$g_r(z)=\int_{\kappa_r}\sum_{n=0}^{\infty}\frac{(sz)^n}{n!}\mathcal{L}(g)(s)\dfrac{ds}{2\pi i}=
\sum_{n=0}^{\infty}\frac{z^n}{n!}\int_{\kappa_r}s^n\mathcal{L}(g)(s)\dfrac{ds}{2\pi i}=
                                                   \sum_{n=0}^{\infty}\frac{a_n}{n!}z^n\; ,$$
where
$$a_n:=\int_{\kappa_r}s^n\mathcal{L}(g)(s)\dfrac{ds}{2\pi i}$$
and we have used uniform convergence. Now, for $|z|>r$ we have,
$$B(g_r)(z)=
\sum_{n=0}^{\infty}\frac{a_n}{z^{n+1}}=
\sum_{n=0}^{\infty}\int_{\kappa_r}\dfrac{1}{z}\left( \dfrac{s}{z}\right)^n\mathcal{L}(g)(s)\dfrac{ds}{2\pi i} =
                                                      \int_{\kappa_r}\dfrac{1}{z-s}\mathcal{L}(g)(s)\dfrac{ds}{2\pi i}\; .$$
This calculation means that the analytic continuation of the Borel transform for $g_r$ is
$$\mathcal{B}(g_r)(z)= \int_{\kappa_r}\dfrac{1}{z-s}\mathcal{L}(g)(s)\dfrac{ds}{2\pi i}=
                                                    \int_{\mathbb{C}}K(z-s)d\mu_r(s)=K*\mu_r(z)\; ,$$
which is an analytic function for every
$z\in \mathbb{C} - \kappa_r$. As a by product we have that the conjugate diagram of $\mathcal{B}(g_r)$ is the convex hull
of the contour $\kappa_r$. Moreover, this means that the type of the function $g_r$ must be $\tau_{g_r}\geq r$, so that
by using the calculus in Item {\sl 1} we conclude that
$\tau_{g_r} = r$. This completes the proof of Item {\sl 1}. Finally we note that the description of the Borel transform
of $g_r$ implies that $g_r$ is recovered via $\mathcal{P}$-Transform from the measure
$$d\mu_{g_r}(s)=K*\mu_r(s)\frac{ds}{2\pi i} \; .$$
\end{proof}

\subsubsection{The truncated equation.}
In this subsection we consider the following "truncated" equation
\begin{equation} \label{Treq}
\zeta(\partial^2_t+h)f_r=g_r\; , \quad \; \; h>1\; ,
\end{equation}
for each $r>0$. We note that in the case $h>1$, the poles of the function are $i\sqrt{h-1}$ and $-i\sqrt{h-1}$, and
therefore we analyse Equation (\ref{Treq}), in the domain
$$\Omega := \mathbb{C} \setminus \{s\in \mathbb{C}: Re(s)\geq 0, \; \; |Im(s)|=\sqrt{h-1} \}\; ,$$
which was used in subsection 4.1.

The following theorem shows that Equation (\ref{Treq}) is well posed in the space $Exp(\Omega)$.

\begin{theorem}\label{TeoSolZ}
 A general solution to Equation $(\ref{Treq})$ in the space $Exp(\Omega)$, is provided by the function
\begin{equation}\label{SolZeta}
\phi_r(z):=\int_{\gamma'}e^{sz}\dfrac{K*\mu_r(s)}{\zeta(s^2+h)}\frac{ds}{2\pi i}=
\int_{\kappa_r}e^{sz}\dfrac{\mathcal{L}(g)(s)}{\zeta(s^2+h)}\dfrac{ds}{2\pi i} +\sum_{j=1}^{N_r}p_j(z)e^{\tau_jz}\; ,
\end{equation}
 where $\gamma' \in H_{1}(g_r)$ is such that it encloses the zeros  $\{\tau_j, j=1,2, \cdots,N_r\}$ 
 of the function $\zeta(s^2+h)$ which lie in the closed ball $\overline{B}_r(0)$, and $p_j(z)$ are polynomials of degree $ord(\tau_j)-1$.
\end{theorem}
\begin{proof}
By Theorem \ref{TheSol} we know that a solution for the Equation (\ref{Treq}) is
$$\int_{\gamma'}e^{sz}\dfrac{\mathcal{B}(g_r)(s)}{\zeta(s^2+h)}\frac{ds}{2\pi i}=
\int_{\gamma'}e^{sz}\dfrac{K*\mu_r(s)}{\zeta(s^2+h)}\frac{ds}{2\pi i} \; ,$$
where $\gamma'$ is the curve in the following picture

\begin{center}
\includegraphics[width=6cm, height=5cm]{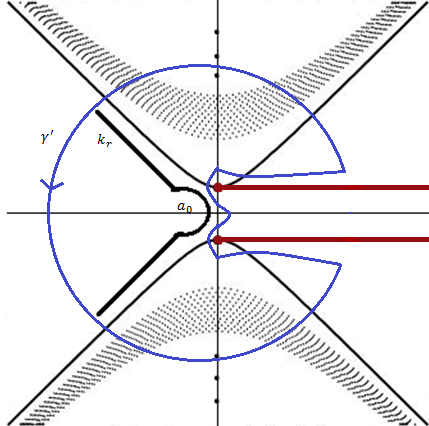}
\end{center}

Since the conjugate diagram $S$ for $\mathcal{B}(g_r)(z)$ is the convex hull of the contour $\kappa_r$, we can decompose the circle
$\{z:|z|'=r\}$ into three pieces $\gamma_1,\gamma_2,\gamma_3$ in which $\gamma_1$ is in the region of analyticity of $\zeta(s^2+h)$ and contains the set $S$ in its interior, while the other two closed paths contain the zeros of $\zeta(s^2+h)$ in its interior, as in the following picture
\begin{center}
\includegraphics[width=6cm, height=5cm]{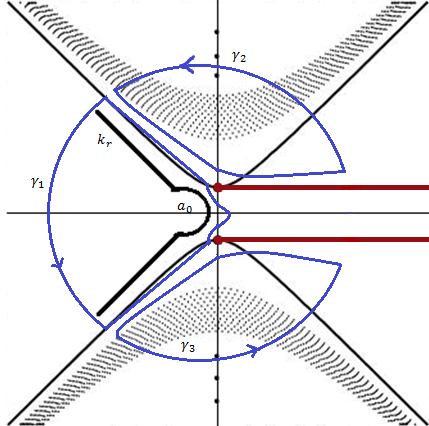}
\end{center}
Therefore,
\begin{eqnarray}\label{3Int}
\phi_r(z)=\int_{|s|'=r}e^{sz}\dfrac{K*\mu_r(s)}{\zeta(s^2+h)}\frac{ds}{2\pi i}&=&\int_{\gamma_1}e^{sz}\dfrac{K*\mu_r(s)}{\zeta(s^2+h)}\frac{ds}{2\pi i}+\int_{\gamma_2}e^{sz}\dfrac{K*\mu_r(s)}{\zeta(s^2+h)}\frac{ds}{2\pi i}+\nonumber \\
&+&\int_{\gamma_3}e^{sz}\dfrac{K*\mu_r(s)}{\zeta(s^2+h)}\frac{ds}{2\pi i}\; .
\end{eqnarray}

\noindent We compute the first integral. Using Fubini's Theorem and the Cauchy integral formula, we obtain
\begin{eqnarray*}
\int_{\gamma_1}e^{sz}\dfrac{K*\mu_r(s)}{\zeta(s^2+h)}\frac{ds}{2\pi i}&=&\int_{\gamma_1}\dfrac{e^{sz}}{\zeta(s^2+h)}\int_{\kappa_r}\dfrac{1}{s-\omega}\mathcal{L}(g)(\omega)\dfrac{d\omega}{2\pi i}\frac{ds}{2\pi i}\\
&=&\int_{\kappa_r}\mathcal{L}(g)(\omega)\int_{\gamma_1}\dfrac{e^{sz}}{\zeta(s^2+h)}\dfrac{1}{s-\omega}\frac{ds}{2\pi i}\dfrac{d\omega}{2\pi i}\\
&=&\int_{\kappa_r}\dfrac{e^{\omega z}}{\zeta(\omega ^2+h)}\mathcal{L}(g)(\omega)\dfrac{d\omega}{2\pi i}\; .
\end{eqnarray*}
\noindent Now the second integral. Using Fubini's theorem again we have
\begin{eqnarray*}
\int_{\gamma_2}e^{sz}\dfrac{K*\mu_r(s)}{\zeta(s^2+h)}\frac{ds}{2\pi i}&=&\int_{\kappa_r}\mathcal{L}(g)(\omega)\int_{\gamma_2}\dfrac{e^{sz}}{\zeta(s^2+h)}\dfrac{1}{(s-\omega)}\frac{ds}{2\pi i}\dfrac{d\omega}{2\pi i}\; ,
\end{eqnarray*}
but now we cannot apply Cauchy's integral formula as before, since the zeros of $\zeta(s^2+h)$ are now poles of the function
\begin{equation}\label{PolFunct}
F(s)=\dfrac{e^{sz}}{\zeta(s^2+h)}\; \dfrac{1}{(s-\omega)}\; ;
\end{equation}
but, we can use the Residue Theorem. Let $\tau_j$ be a zero of the function $\zeta(s^2+h)$ lying inside the region enclosed
by the curve $\gamma_2$. We have,
$$Res (F,\tau_j)=\sum_{l=0}^{ord(\tau_j)-1}h_l(\omega,\tau_j)z^le^{\tau_j z}\; ,$$
for some functions $h_l$. Now we let $N_{2,r}$ be the number of zeros of the function $\zeta(s^2+h)$ inside the region enclosed by the curve $\gamma_2$. We conclude that the second integral becomes

\begin{eqnarray*}
\int_{\gamma_2}e^{sz}\dfrac{K*\mu_r(s)}{\zeta(s^2+h)}\frac{ds}{2\pi i}&=&\int_{\kappa_r}\mathcal{L}(g)(\omega)\sum_{j=1}^{N_{2,r}} \left( \sum_{l=0}^{ord(\tau_j)-1}h_l(\omega,\tau_j)z^le^{\tau_j z}\right) \dfrac{d\omega}{2\pi i}\\
&=&\sum_{j=1}^{N_{2,r}} \sum_{l=0}^{ord(\tau_j)-1} z^le^{\tau_j z} \int_{\kappa_r}\mathcal{L}(g)(\omega)h_l(\omega,\tau_j)\dfrac{d\omega}{2\pi i}\\
&=& \sum_{j=1}^{N_{2,r}} \sum_{l=0}^{ord(\tau_j)-1} A_l(\tau_j) z^le^{\tau_j z}\\
&=&\sum_{j=1}^{N_{2,r}}a_j(z)e^{\tau_j z}\; ,
\end{eqnarray*}
where we have defined the polynomials
$$a_j(z):=\sum_{l=0}^{ord(\tau_j)-1} A_l(\tau_j) z^l\; .$$

\noindent Finally, let $N_{3,r}$ be the number of zeros
of the function $\zeta(s^2+h)$ inside the region enclosed by the curve $\gamma_3$. We use the same strategy as above for the third integral in (\ref{3Int}) and we obtain
\begin{eqnarray*}
\int_{\gamma_3}e^{sz}\dfrac{K*\mu_r(s)}{\zeta(s^2+h)}\frac{ds}{2\pi i}=\sum_{j=1}^{N_{3,r}}b_j(z)e^{\tau_j z},
\end{eqnarray*}
Putting
$N_r=N_{2,r}+N_{3,r}$ as the number of zeros inside of the closed ball $\overline{B}_r(0)$, and setting $p_j= a_j$ for $j=1,2, \cdots , N_{2,r}$ and $p_j= b_j$ for $j=1,2, \cdots , N_{3,r}$ we obtain equality (\ref{SolZeta}) and the theorem is proved.
\end{proof}

In what follows we consider only the particular solution
\begin{equation}\label{ForParSol}
\phi_r(z)=\int_{\kappa_r}e^{sz}\dfrac{\mathcal{L}(g)(s)}{\zeta(s^2+h)}\dfrac{ds}{2\pi i}\;
\end{equation}
\noindent
to Equation (\ref{Treq}). This solution is obtained from Theorem \ref{TeoSolZ} by using a curve $\gamma_1$ as in the following picture
\begin{center}
\includegraphics[width=6cm, height=5cm]{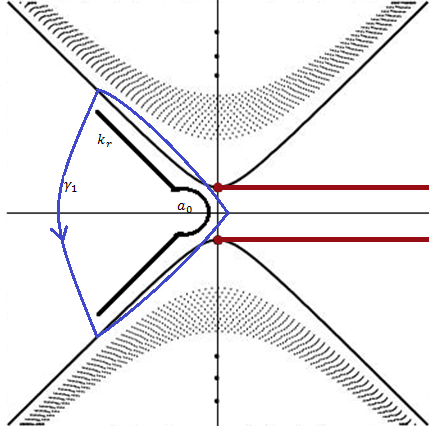}
\end{center}

One reason for considering only this expression is that the contribution of the second summand in (\ref{SolZeta})
``can be ignored'', since it corresponds to a solution of the homogeneous equation $\zeta(\partial_t^2+h)f_r=0$. Also, we
note that it is still an open problem whether the zeros of the Riemann Zeta function are simple or not (see for example
\cite{Anderson83,Bui13,Cheer93}); consequently, we do not even know a precise upper bound for the degree of the polynomials
$p_j$ appearing in Theorem (\ref{TeoSolZ})! Such an information could be used, for example,
for the study of the uniform convergence of the sequence of partial sums determined for the second summand in (\ref{SolZeta})
for each $r$.

\begin{rem}
On the other hand, from {\rm \cite{Brend79,Brend82,VdeLune83,VdeLune86}}, we know that the first zeros of the Riemann zeta
function are simple; therefore the first zeros of $\zeta(s^2+h)$ are also simple. Let $r>0$ and suppose that the curve
$\gamma' \in H_1(g_r)$ encloses the first known simple zeros of $\zeta(s^2+h)$; then, in this situation the full
representation formula for the solution given in Theorem \ref{TeoSolZ} is more concrete. This situation is treated in
the example that follows.
\end{rem}

\begin{example}
From the work \cite{VdeLune86} (and references therein) we know that at least the first $1.500.000.001$ zeros of the Riemann
Zeta function are simple and are located at the critical line; therefore the first zeros of $\zeta(s^2+h)$ are also simple.
This implies that the first terms of the sequence of sums in the representation formula (\ref{SolZeta}) are easy to calculate.

In fact, let $r>0$ be such that the curve $\gamma' \in H_1(g_r)$ encloses the first  $3.000.000.002$ simple zeros of the
shifted  Riemann Zeta function $\zeta(s^2+h)$. \\
Let $j\in \{1,2,3,\cdots ,3.000.000.002\}$ and let $\tau_j$ be the corresponding simple zero. If we define
$$\zeta_j=\lim_{s\to \tau_j}\dfrac{\zeta(s^2+h)}{s-\tau_j}\; ,$$
then, applying the residue theorem to the function $F$ defined in Equation (\ref{PolFunct}) we obtain
$$Res(F,\tau_j)=\dfrac{e^{\tau_jz}}{\zeta_j (\tau_j-\omega)}\; .$$
Therefore, from the proof of Theorem \ref{TeoSolZ} we have that the representation formula of the solution is reduced to
$$\phi_r(z)=\int_{\kappa_r}e^{sz}\dfrac{\mathcal{L}(g)(s)}{\zeta(s^2+h)}\dfrac{ds}{2\pi i}
 +\sum_{j=1}^{N_r}c_je^{\tau_jz}\; ,$$
where $c_j$ are the following complex numbers:
$$c_j:=\dfrac{1}{\zeta_j}\int_{\kappa_r}\dfrac{\mathcal{L}(g)(\omega)}{\zeta_j(\tau_j-\omega)}\dfrac{d\omega}{2\pi i}\; .$$
%
%
\end{example}

\subsubsection{A particular solution}
The proof of the following lemma can be found in \cite[Theorem 36.1]{Doetsch}
\begin{lemma}\label{Lemmm}
Let $g \in \mathcal{L}_{>}(\mathbb{R}_+)$ and let $\kappa$ be the angular contour of the domain of the analytic extension of $\mathcal{L}(g)$ with centre $a_0=0$ and half-angle of opening $\psi$, where $\frac{\pi}{2}<\psi\leq \pi$. Then, the function
$$g_{\infty}(z):=\int_{\kappa_{\infty}}e^{zs}\mathcal{L}(g)(s)\frac{ds}{2\pi i},$$
is analytic in an angular region with horizontal bisector and half-angle of opening $\psi-\frac{\pi}{2}$.
\end{lemma}

\noindent Let us denote by $D_{\psi}$ the angular region with horizontal bisector and half-angle of opening $\psi-\frac{\pi}{2}$  arising in the previous lemma, see  \cite[figure 32, p. 243]{Doetsch}. We note that $g_{\infty}$ is analytic in $D_{\psi}$. We can estimate $\psi$.


We can see that the real functions $y=|x|$ and $y=-|x|$ are asymptotes to the region which contain the zeros of $\zeta(s^2+h)$. Therefore, the angle $\psi$ satisfies  $\frac{3\pi}{4}< \psi \leq \pi$. This gives us a natural fixed angular region $D_{\frac{3\pi}{4}}$ on which the function $g_{\infty}$ is analytic, since $D_{\frac{3\pi}{4}}\subset D_{\psi}$ for all $\frac{3\pi}{4}< \psi \leq \pi$.

\begin{proposition} \label{lim} 
Let $a_0=0$  be the first singularity of the analytic extension of $\mathcal{L}(g)$ and also let $\frac{3\pi}{4}< \psi \leq \pi$ be the angle
 described in lemma {\rm \ref{Lemmm}}. Then, on compact subsets of $D_{\psi} \subset \mathbb{C}$ we have:
 \begin{enumerate}
  \item The sequence $\{g_r\}_{r>0}$ converge uniformly to 
  $$g_{\infty}(z)=\int_{\kappa_{\infty}}e^{zs}\mathcal{L}(g)(s)\frac{ds}{2\pi i}\; .$$

  \item The sequence $\{f_r\}_{r>0}$ converge uniformly to
$$f_{\infty}(z):=\int_{\kappa_{\infty}}e^{sz}\dfrac{\mathcal{L}(g)(s)}{\zeta(s^2+h)}\dfrac{ds}{2\pi i}\; .$$
 \end{enumerate}
 In particular both conclusions hold on $D_{\frac{3\pi}{4}}$.
\end{proposition}
\begin{proof}
 We prove Item {\sl 1}. Let $K$ be a compact subset of $D_{\psi}$; since it is closed and bounded, there is a positive number
 $\delta$ such that the distance between $D_{\psi}$ and $\partial K$ (the topological boundary of $K$) is at least $\delta$.
Also, there exist a positive number $A$ and angles $\theta_1,\; \theta_2$ satisfying
 $\frac{\pi}{2}-\psi<\theta_1<\theta_2<\psi-\frac{\pi}{2}$, such that for all $z\in K$
\begin{enumerate}
\item[a.] $|z|\geq A$, and
\item[b.] $\theta_1<\theta_z<\theta_2$, where $\theta_z$ denotes the angle of $z$ with the real line, $z = |z| \exp(i \theta_z)$.
\end{enumerate}
Therefore, for $z\in K$ we have

\begin{eqnarray*}
 |g_r(z)-g_{\infty}(z)|&=&\left|\int_{\kappa_{\infty}-\kappa_r}e^{sz}\mathcal{L}(g)(s)\dfrac{ds}{2\pi i}\right|\\
&\leq &\left|\int_r^{\infty}e^{te^{i\psi}z}\mathcal{L}(g)(te^{i\psi})e^{i\psi}\dfrac{dt}{2\pi}\right|+\left|\int_r^{\infty}e^{te^{-i\psi}z}\mathcal{L}(g)(te^{-i\psi})e^{-i\psi}\dfrac{dt}{2\pi}\right|.
\end{eqnarray*}
Let $l_z=|z|$; for the first integral, we have

\begin{eqnarray*}
\left|\int_r^{\infty}e^{te^{i\psi}z}\mathcal{L}(g)(te^{i\psi})e^{i\psi}\dfrac{dt}{2\pi}\right|&=&\left|\int_r^{\infty}e^{tl_ze^{i(\psi+\theta_z)}}\mathcal{L}(g)(te^{i\psi})e^{i\psi}\dfrac{dt}{2\pi}\right|\\
&\leq&\int_r^{\infty}e^{tl_z\cos(\psi+\theta_z)}\left|\mathcal{L}(g)(te^{i\psi})\right|\dfrac{dt}{2\pi}\; .
\end{eqnarray*}
By the Riemann-Lebesgue Lemma we have that $\mathcal{L}(g)$ is bounded, and therefore
$|\mathcal{L}(g)(te^{i\psi})|\leq M_{\mathcal{L}(g)}$ for $t\geq r$. Also, $\frac{\pi}{2}<\theta_1+\psi<\psi+\theta_z<\theta_2+\psi<\frac{3\pi}{2}$, which implies that $\cos(\psi+\theta_z)<-B<0$,
for some $B>0$. Therefore, we have

\begin{eqnarray*}
\int_r^{\infty}e^{tl_z\cos(\psi+\theta_z)}\left|\mathcal{L}(g)(te^{i\psi})\right|\dfrac{dt}{2\pi}
&\leq& M_{\mathcal{L}(g)}\int_r^{\infty}e^{tl_z\cos(\psi+\theta_z)}\dfrac{dt}{2\pi}\\
&\leq&\int_r^{\infty}e^{-tl_zB}\dfrac{dt}{2\pi}=\dfrac{1}{l_zB}e^{-rl_zB}\\
&\leq&\dfrac{1}{AB}e^{-rAB}\; .
\end{eqnarray*}

For the second integral, we have $-\frac{3\pi}{2}<\theta_1-\psi<\theta_z-\psi<\theta_2-\psi<-\frac{\pi}{2}$, and therefore there is a constant $C>0$
such that $\cos(\theta_z-\psi)<-C<0$. Then
\begin{eqnarray*}
\left|\int_r^{\infty}e^{te^{-i\psi}z}\mathcal{L}(g)(te^{-i\psi})e^{-i\psi}\dfrac{dt}{2\pi} \right|
& = & \int_r^{\infty}e^{tl_z\cos(\theta_z-\psi)}\left|\mathcal{L}(g)(te^{-i\psi})\right|\dfrac{dt}{2\pi}\\
&\leq& M_{\mathcal{L}(g)}\int_r^{\infty}e^{tl_z\cos(\theta_z-\psi)}\dfrac{dt}{2\pi}\\
&\leq&\int_r^{\infty}e^{-tl_zC}\dfrac{dt}{2\pi}=\dfrac{1}{l_zC}e^{-rl_zC}\\
&\leq&\dfrac{1}{AC}e^{-rAC}\; .
\end{eqnarray*}
These computations allow us to conclude that given $\epsilon>0$ there is $r_0>0$ such that for every $r>r_0$ and for every  $z\in K$ we have the estimate:
\begin{eqnarray*}
|g_r(z)-g_{\infty}(z)|&\leq& \dfrac{1}{AB}e^{-rAB}+\dfrac{1}{AC}e^{-rAC}<\epsilon\; .
\end{eqnarray*}

Item 2 follows from the fact that the function $\dfrac{1}{\zeta(s^2+h)}$ is bounded on the angular contour $\kappa_{\infty}$ for $|s|\to \infty$.
\end{proof}

Let us now denote the angle $\psi$ described in Lemma \ref{Lemmm} by $\psi(g)$. 
From the result in Proposition \ref{lim} we have the following remark
\begin{rem} \label{EndRem} We have
\begin{enumerate}
\item Proposition {\rm \ref{lim}} implies that the function $g_{\infty}$ is an anlytic function which extends $g$; that is
$g_{\infty}(t)=g(t) \; \forall t \in \mathbb{R}_+$.
\item The sequences $\{f_r\}_{r>0}$ and $\{g_r\}_{r>0}$ are sequences of entire functions of increasing exponential type $r$. On the other hand,
 functions $f_{\infty} $ and $g_{\infty}$ from Proposition {\rm \ref{lim}} are, generally speaking, neither entire nor of finite exponential type.
\item In principle the functions $g_{\infty}$ and $f_{\infty}$ depend on $\psi$, with $\frac{3\pi}{4}<\psi \leq \psi(g)$: for each angle $\psi$ in
$ ]\frac{3\pi}{4},\psi(g)]$ and each $r>0$, we obtain the finite angular contour $\kappa^{\psi}_r$ (which is part of the infinite angular contour
$\kappa^{\psi}$), the sequence of functions $\{f^{\psi}_r\}_{r>0}$ and $\{g^{\psi}_r\}_{r>0}$, and the limit functions $ g^{\psi}_{\infty}$ and
$ f^{\psi}_{\infty}$. We also note that for $\psi_1\leq\psi_2$ in $]\frac{3\pi}{4},\psi(g)]$ the functions $g^{\psi_2}_{\infty}$ and
$f^{\psi_2}_{\infty}$ are analytic extensions of $g^{\psi_1}_{\infty}$ and $f^{\psi_1}_{\infty}$ respectively.
\end{enumerate}
\end{rem}

Motivated by this remark and Proposition \ref{lim}, we define the following nonempty set:

\begin{eqnarray*}
\mathcal{W}(g):=\bigg \{ f^{\psi}_{\infty} \; : \; \psi  \in \; ] \frac{3\pi}{4},\psi(g)]\bigg \} \;  .
\end{eqnarray*}
Also, we denote by $\Omega_{\frac{3\pi}{4}}$ the reflexion of $D_{\frac{3\pi}{4}}$ with respect to the imaginary axis. We define the operator
$\widetilde{\zeta}(\partial_t^2+h)$ on $\mathcal{W}(g)$ as follows:

\begin{defi}\label{DefFin}
Let $f^{\psi}_{\infty} \in \mathcal{W}(g)$
and let $f^{\psi}_r \in Exp(\Omega_{\frac{3\pi}{4}})$ be a family which satisfies Equation $(\ref{Treq})$ and such that
$f^{\psi}_r \to f^{\psi}_{\infty}$ in the the topology of uniform convergence on compact subsets of $Dom(f^{\psi}_{\infty})$. Then,
\begin{equation}\label{EqEq}
\widetilde{\zeta}(\partial_t^2+h)f^{\psi}_{\infty} := \lim_{r\to\infty}\zeta(\partial_t^2+h)f^{\psi}_r \; ,
\end{equation}
where the limit is also taken in the topology of uniform convergence on compact subsets of $Dom(f^{\psi}_{\infty})$.
\end{defi}
Because of \cite[Theorem 25.1]{Doetsch} the limit appearing in the right hand side of Equation (\ref{EqEq}) {\em does not depend} on the choice
of the angle $\psi$. By the same reason the function $f^{\psi}_{\infty}$ {\em does not depend} on $\psi$. Thus we can use
%
Definition \ref{DefFin} to interpret Equation (\ref{ExEq}) in the case in which the data $g\in \mathcal{L}_{>}(\mathbb{R}_+)$: we look, for a fixed $\psi$, a solution $f^{\psi}_{\infty}$ in the set $ \mathcal{W}(g)$ to the following equation:
\begin{equation}\label{FiEqn}
\widetilde{\zeta}(\partial_t^2+h)f^{\psi}_{\infty}=g\; ,
\end{equation}
and we understand Equation (\ref{FiEqn}) in the following limit sense:
%
$$\lim_{r\to\infty}\zeta(\partial_t^2+h)f^{\psi}_r =\lim_{r\to\infty}g^{\psi}_r = g^{\psi}_{\infty}\; ,$$
where the sequences $\{f^{\psi}_r\}_{r>0}$ and $\{g^{\psi}_r\}_{r>0}$ are in $Exp(\Omega_{\frac{3\pi}{4}})$ and they are related as in Proposition
\ref{lim}. We recall once more that limit is taken under the topology of uniform convergence on compact subsets of $Dom(f^{\psi}_{\infty})$, and
that $g^{\psi}_{\infty}$ do not depend on the angle $\psi$ (again because of \cite[Theorem 25.1]{Doetsch},).

\begin{proposition}
Let us consider the particular angle $\psi=\psi(g)$ defined after Proposition {\rm \ref{lim}}. The solution to Equation {\rm (\ref{FiEqn})} is
the function $f^{\psi(g)}_{\infty}\in  \mathcal{W}(g)$ given in Proposition {\rm \ref{lim}}.
\end{proposition}
\begin{proof} From Proposition \ref{lim}, we recall that
$$f^{\psi(g)}_{\infty}(z)= \int_{\kappa^{\psi(g)}}e^{sz}\dfrac{\mathcal{L}(g)(s)}{\zeta(s^2+h)}\dfrac{ds}{2\pi i}\; ,$$
and that there exist a function $g^{\psi(g)}_{\infty}$ given by
%
$$g^{\psi(g)}_{\infty}(z)=\int_{\kappa^{\psi(g)}}e^{zs}\mathcal{L}(g)(s)\frac{ds}{2\pi i}\; .$$
on the domain $ Dom(f^{\psi(g)}_{\infty})$. The analytic function $g^{\psi(g)}_{\infty}$ extends the function $g$ defined in
principle on $\mathbb{R}_+$.

Furthermore, there exist explicit sequences $\{f^{\psi(g)}_r\}_{r>0}$ and $\{g^{\psi(g)}_r\}_{r>0}$ in $Exp(\Omega_{\frac{3\pi}{4}})$ given by:
$$f_{r}^{\psi(g)}(z)=\int_{\kappa^{\psi(g)}_r}e^{sz}\dfrac{\mathcal{L}(g)(s)}{\zeta(s^2+h)}\dfrac{ds}{2\pi i}\; ,$$
and
$$g_r^{\psi(g)}(z)=\int_{\kappa^{\psi(g)}_r}e^{sz}\mathcal{L}(g)(s)\dfrac{ds}{2\pi i}\; .$$
%
%
These sequences, for each $r>0$, satisfy the following truncated equations on $Exp(\Omega_{\frac{3\pi}{4}})$

\begin{equation}\label{FInFIn}
\zeta(\partial_t^2+h)f^{\psi(g)}_r=g^{\psi(g)}_r\; .
\end{equation}
Furthermore, in Proposition \ref{lim} we proved that on compact subsets of $ Dom(f^{\psi(g)}_{\infty})$, the following two uniform limits holds
\begin{enumerate}
 \item [a).] $$\lim_{r\to \infty}g^{\psi(g)}_r(z)=g^{\psi(g)}_{\infty}(z)\; ,$$ 
%
 \item [b).] $$\lim_{r\to \infty}f^{\psi(g)}_r(z)= f^{\psi(g)}_{\infty}(z)\; .$$
\end{enumerate}
Therefore, taking limits in Equation (\ref{FInFIn}) and using items a) and b), the following equality hold (on $Dom(f^{\psi(g)}_{\infty})$)
$$\lim_{r\to \infty} \zeta(\partial_t^2+h)f^{\psi(g)}_r(z)=\lim_{r\to \infty}g^{\psi(g)}_r(z)=g^{\psi(g)}_{\infty}(z)\; .$$
That is, on $Dom(f^{\psi(g)}_{\infty})$ we have
$$\widetilde{\zeta}(\partial_t^2+h)f^{\psi(g)}_{\infty}=g^{\psi(g)}\; . $$
In particular
$$\widetilde{\zeta}(\partial_t^2+h)f^{\psi(g)}_{\infty}(t)=g(t) \; \; in \; \; \mathbb{R}_+ \; .$$
\end{proof}

\section*{Appendix: Some Zeta-nonlocal scalar fields}

\subsection{Equations of motion}
Following Dragovich's work \cite{D}, we show how to deduce several mathematical interesting nonlocal scalar field
equations whose dynamics depends on the Riemann zeta function, Hurwitz-zeta function and also on a Dirichlet-Taylor series.


Recall that, given a prime number $p$, the Lagrangian formulation of the open $p-$adic
string tachyon is
\begin{equation}\label{Zeq_03}
L_p=\dfrac{m_p^D}{g_p^2}\dfrac{p^2}{p-1}\big{(-}\dfrac{1}{2}\phi p^{-\square/(2m_p^2)}\phi+\dfrac{1}{p+1}\phi^{p+1}\big{)}\; ,
\end{equation}
where $\square$ is the D'Alembert operator defined by $\, \square:=-\partial_t^2+ \triangle_x$, in which $\triangle_x$ is
the Laplace operator and we are using metric signature $(-,+,\cdots,+)$, following \cite{D}. This Lagrangian is defined only
formally; as we have shown here, the terms appearing therein are well-defined in the
$1+0$ case, see also \cite{CPR_Laplace,CPR,GPR_CQG}.
The equation of motion for (\ref{Zeq_03}) is
$$p^{-\square/(2m_p^2)}\phi=\phi^p.$$
Dragovich has considered the model
$$L=\sum_{n=1}^{\infty}C_nL_n=
\sum_{n=1}^{\infty}C_n\dfrac{m_n^D}{g_n^2}\dfrac{n^2}{n-1}\big{(-}\dfrac{1}{2}\phi n^{-\square/(2m_n^2)}\phi +
\dfrac{1}{n+1}\phi^{n+1}\big{)}\; ,$$
in which all lagrangians $L_n$ given by (\ref{Zeq_03}) are taken into account. Explicit lagrangians $L$ depend on the choices of the coefficients $C_n$. Some particular cases are considered bellow.

\subsubsection{The Riemann zeta function as a symbol}
This is the case in \cite{D} and one of our main motivations. We recall once again that the Riemann zeta function is
defined by (see for instant \cite{KaVo})
$$\zeta(s):=\sum_{n=1}^{\infty}\dfrac{1}{n^s}\; , \; \; \quad Re(s)>1\; .$$
It is analytic on its domain of definition and it has an analytic extension to the whole complex plane with the exception of
the point $s=1$, at which it has a simple pole with residue $1$.

If we consider the explicit coefficient
$$C_n= \dfrac{n-1}{n^{2+h}}\, ,$$
in which $h$ is a real number, Dragovich's Lagrangian becomes
$$ L_h=\dfrac{m^D}{g^2}\left(-\dfrac{1}{2}\phi \sum_{n=1}^{\infty}n^{-\square/(2m_n^2)-h}\phi + \sum_{n=1}^{\infty}\dfrac{n^{-h}}{n+1}\phi^{n+1}\right)\, .$$
We write $L_h$ in terms of the zeta function and, in order to avoid convergence issues, we replace the nonlinear term
for an adequate analytic function $G(\phi)$. The Lagrangian $L_h$ becomes:
$$L_h= \dfrac{m^D}{g^2}\left( -\dfrac{1}{2}\phi \zeta(\dfrac{\square}{2m^2}+h)\phi + G(\phi) \right)\; .$$
The equation of motion is
$$\zeta(\dfrac{\square}{2m^2}+h)\phi=g(\phi)\; ,$$
in which $g = G'$.

\subsubsection{Dirichlet zeta function as symbol}

Let us consider $\chi$ a Dirichlet character modulo $m$ and let us define
$$C_n =\dfrac{\chi(n)(n-1)}{n^{2+h}}\; .$$
We recall that a $L$-Dirichlet series is of the following form:
$$L(s,\chi)=\sum_{n=1}^{\infty}\dfrac{\chi(n)}{n^s}$$
%
following Dragovich's approach, we can consider the Lagrangian
$$L_h= \dfrac{m^D}{g^2}\left( -\dfrac{1}{2}\phi L(\dfrac{\square}{2m^2}+h,\chi)\phi + F(\varphi) \right) $$
and the corresponding equation of motion
\begin{equation}
L(\dfrac{\square}{2m^2}+h,\chi)\phi=f(\phi)\; ,
\end{equation}
in which $f = F'$.

\subsubsection{Almost periodic Dirichlet series as symbol}

Let $\{a_n\}$ be a sequence of complex numbers. A Dirichlet series is a series of the form
$$F(s):=\sum_{n=1}^{\infty}\dfrac{a_n}{n^s}\; .$$
Then, for a given sequence $\{a_n\}$, if we consider the coefficients
$$C_n=\dfrac{a_n(n-1)}{n^{2+h}},$$
we arrive at the following Lagrangian and equation of motion:
$$L_h= \dfrac{m^D}{g^2} \left( -\dfrac{1}{2}\phi F(\dfrac{\square}{2m^2}+h)\phi + D(\phi) \right)\; ,$$
\begin{equation}\label{eqDiSe}
F(\dfrac{\square}{2m^2}+h)\phi=d(\phi)\; ,
\end{equation}
in which $d = D'$.

A particular case of this equation is the equation with dynamics depending on Dirichlet series with {\sl almost periodic
coefficients}: following \cite{OKnill}, we consider a piecewise continuous, $1$-periodic and  $L^2$-function
$f: \mathbb{R}\to \mathbb{C}$ with Fourier expansion $f(x)= \sum_{k=-\infty}^{\infty}b_ke^{2\pi i kx}$; the particular
symbol of interest for equation (\ref{eqDiSe}) is the following {\sl almost periodic Dirichlet series}:

$$F_{\alpha}(s):= \sum_{n=1}^{\infty}\dfrac{f(n\alpha)}{n^s}\; .$$

\paragraph{Acknowledgments}

A.C. has been supported by PRONABEC (Ministerio de Educaci\'on, Per\'u) and FONDECYT
through grant \# 1161691; H.P. and E.G.R. have been
partially supported by the FONDECYT operating grants \# 1170571 and
\# 1161691 respectively.

\end{document}